\documentclass[11pt]{article}
\usepackage{fullpage}
\usepackage{amsthm}
\usepackage{amsmath,amsfonts,amssymb}
\usepackage{color}
\usepackage[hidelinks]{hyperref}
\usepackage{smartdiagram}

\usepackage[boxruled,titlenotnumbered]{algorithm2e}
\LinesNumberedHidden

\makeatletter
\newcommand{\RemoveAlgoNumber}{\renewcommand{\fnum@algocf}{\AlCapSty{\AlCapFnt\algorithmcfname}}}
\newcommand{\RevertAlgoNumber}{\algocf@resetfnum}
\makeatother
\RemoveAlgoNumber

\newenvironment{adversary}[1][htb]
{\begin{center}
\begin{minipage}{.6\linewidth}
  \NoCaptionOfAlgo
  \RestyleAlgo{boxed}
  \renewcommand{\algorithmcfname}{}
   \begin{algorithm}[#1]%
     }{\end{algorithm}\end{minipage}\end{center}}

\newenvironment{mechanism}[1][htb]
{\begin{center}
\begin{minipage}{.6\linewidth}
\RestyleAlgo{boxruled}
  \renewcommand{\algorithmcfname}{Mechanism}
   \begin{algorithm}[#1]%
  }{\end{algorithm}\end{minipage}\end{center}}

\newcommand{\supp}{\mathsf{supp}}
\newcommand{\lowset}{\low}
\newcommand{\highset}{\high}
\newcommand{\SD}{\mathsf{SD}}

\newcommand{\bigmid}{\;\big\vert\;}
\newcommand{\Bigmid}{\;\Big\vert\;}

\newcommand{\mech}{\mathsf{M}}

\newcommand{\row}{x}
\newcommand{\univ}{X}

\newcommand{\bits}{\mathsf{bits}}
\newcommand{\weight}{\mathsf{weight}}

\newcommand{\enc}{\mathsf{enc}}

\newcommand{\sk}{\mathsf{s}}
\newcommand{\ct}{\mathsf{c}}
\newcommand{\msg}{\row_n}

\newcommand{\ee}{\mathsf{ExtEnc}}
\newcommand{\ext}{\mathsf{ext}}
\newcommand{\hold}{\mathsf{hold}}

\newcommand{\indic}[1]{\mathbb{I}(#1)}

\newcommand{\adv}{\mathsf{A}}
\newcommand{\A}{\adv}
\newcommand{\T}{\mathsf{T}}
\newcommand{\B}{\mathsf{B}}
\newcommand{\Succ}{\mathsf{Succ}}

\newcommand{\E}{\mathop{\mathbb{E}}}
\renewcommand{\land}{~\wedge~}

\newcommand{\bbN}{\mathbb{N}}
\newcommand{\anon}{\mathsf{Anon}}
\newcommand{\kmax}{k_{\mathsf{max}}}

\newcommand{\iso}[2]{\mathsf{iso}(#1,#2)}

\newcommand{\low}{\mathsf{low}}
\newcommand{\high}{\mathsf{high}}
\newcommand{\w}{w}

\newcommand{\dv}[2]{\frac{d#1}{d#2}}

\newcommand{\base}{\mathsf{base}}

\newcommand{\lsb}{\mathsf{lsb}}

\def\opt{\mathop{\rm{opt}}\nolimits}

\newcommand\N{\mathbb{N}}

\newcommand{\cD}{\mathcal{D}}

\newcommand{\cP}{\mathcal{P}}

\newcommand{\cX}{\mathcal{X}}

\newcommand{\eps}{\varepsilon}

\newcommand{\negl}{\mathrm{negl}}

\newtheorem{theorem}{Theorem}[section]
\newtheorem{lemma}[theorem]{Lemma}
\newtheorem{fact}[theorem]{Fact}
\newtheorem{claim}[theorem]{Claim}
\newtheorem{remark}[theorem]{Remark}
\newtheorem{corollary}[theorem]{Corollary}
\newtheorem{proposition}[theorem]{Proposition}

\theoremstyle{definition}
\newtheorem{definition}{Definition}[section]

\newcommand{\db}{\mathbf{x}}
\newcommand{\dby}{\mathbf{y}}
\newcommand{\cnt}{\mathsf{count}}

\def\E{\operatorname*{\mathbb{E}}}

\def\opt{\mathop{\rm{opt}}\nolimits}

\newcommand{\zo}{\{0,1\}}

\newcounter{notecounter}

\newcommand{\mynote}[1]{\stepcounter{notecounter}\textcolor{red}{\tiny \bf [\arabic{notecounter}]}\marginpar{\tiny \sf \textcolor{red}{\bf [\arabic{notecounter}]}#1}}

\newcommand{\todo}[1]{\textcolor{red}{\sf [TODO: {#1}]}}

\usepackage[normalem]{ulem} 

\iffalse
\else 
\renewcommand{\mynote}[1]{}
\renewcommand{\todo}[1]{}
\fi

\title{Towards Formalizing the GDPR's Notion of Singling Out}

\author{Aloni Cohen\thanks{MIT. {\tt aloni@mit.edu}.}
\and
Kobbi Nissim\thanks{Department of Computer Science, Georgetown University. {\tt kobbi.nissim@georgetown.edu}.}}

\date{\today}

\begin{document}

\begin{titlepage}
\maketitle

\thispagestyle{empty}

\begin{abstract}

There is a significant conceptual gap between legal and mathematical thinking around data privacy.
The effect is uncertainty as to which technical offerings adequately match expectations expressed in legal standards.
The uncertainty is exacerbated by a litany of successful privacy attacks demonstrating that traditional statistical disclosure limitation techniques often fall short of the sort of privacy envisioned by legal standards.

We define {\em predicate singling out}, a new type of privacy attack intended to capture the concept of singling out appearing in the General Data Protection Regulation (GDPR).
Informally, an adversary predicate singles out a dataset $\db$ using the output of a data-release mechanism $M(\db)$ if it manages to find a predicate $p$ matching exactly one row $\row\in\db$ with probability much better than a statistical baseline.
A data-release mechanism that precludes such attacks is {\em secure against predicate singling out} (\emph{PSO secure}).

We argue that PSO security is a mathematical concept with legal consequences. Any data-release mechanism that purports to ``render anonymous'' personal data under the GDPR must be secure against singling out, and, hence, must be PSO secure.
We then analyze PSO security, showing that it fails to self-compose.
Namely, a combination of $\omega(\log n)$ exact counts, each individually PSO secure, enables an attacker to predicate single out. In fact, the composition of just two PSO-secure mechanisms can fail to provide PSO security.

Finally, we ask whether differential privacy and $k$-anonymity are PSO secure.
Leveraging a connection to statistical generalization, we show that differential privacy implies PSO security. However, and in contrast with current legal guidance, $k$-anonymity does not: There exists a simple and general predicate singling out attack under mild assumptions on the $k$-anonymizer and the data distribution.

\paragraph{Note:} This is the preliminary version. The journal version \cite{cohen2020towards} incorporates a few corrections and simplifications, but contains fewer details.\footnote{Available at \url{https://www.pnas.org/content/117/15/8344}.}
\end{abstract}

\bigskip

\end{titlepage}

\tableofcontents
\thispagestyle{empty}

\newpage

\setcounter{page}{1}


\section{Introduction}

Data privacy laws---like HIPAA, FERPA, and Title~13 in the US, and the GDPR in the EU---govern the use of sensitive personal information.\footnote{HIPAA is the Health Insurance Portability and Accountability Act. FERPA is the Family Educational Rights and Privacy Act. Title~13 of the US Code mandates the role of the US Census. GDPR is the EU General Data Protection Regulation.}
These laws delineate the boundaries of appropriate use of personal information and impose steep penalties upon rule breakers.
To adhere to these laws, practitioners need to apply suitable controls and statistical disclosure limitation techniques.
Many commonly used techniques including $k$-anonymity, bucketing, rounding, pseudonymization, and swapping offer privacy protections that are seemingly intuitive but only poorly understood.
And while there is a vast literature of best practices, a litany of successful privacy attacks demonstrates that these techniques often fall short of the sort of privacy envisioned by legal standards.\footnote{See, e.g.,~\cite{BrokenPromises}.}

A more disciplined approach is needed.
However, there is a significant conceptual gap between legal and mathematical thinking around data privacy.
Privacy regulations are grounded in legal concepts such as personally-identifiable information (PII), linkage, distinguishability, anonymization, risk, and inference.
In contrast, much of the recent progress in data privacy technology is rooted in mathematical privacy models such as differential privacy~\cite{DMNS06} that offer a foundational treatment of privacy, with formal privacy guarantees.
And while such techniques are being actively developed in the academy, industry, and government, there is a basic disconnect between the legal and mathematical conceptions.
The effect is uncertainty as to which technical offerings adequately match expectations expressed in legal standards~\cite{NW18}.

\paragraph{Bridging between legal and technical concepts of privacy.} We aim to address this uncertainty by translating between the legal and the technical.
To do so, we begin with a concept appearing in the law, then model some aspect of it mathematically.
With the mathematical formalism in hand, we can better understand the \emph{requirements} of the law, their \emph{implications}, and the \emph{techniques} that might satisfy them.

This is part of a larger effort to bridge between legal and technical conceptions of privacy.
An earlier work analyzed the privacy requirements of FERPA and modeled them in a game-based definition, as is common in cryptography.
The definition was used to argue that the use of differentially private analyses suffices for satisfying a wide range of interpretation of FERPA~\cite{Bridging}.
An important feature of FERPA that enabled this analysis is that FERPA and its accompanying documents contain a rather detailed description of a privacy attacker and the attacker's goals.

In this work we focus on the concept of {\em singling out} from the GDPR.
More specifically, we examine what it means for a data anonymization mechanism to ensure \emph{security against singling out} in a data release.
Preventing singling out attacks in a dataset is a necessary (but maybe not sufficient) precondition for a dataset to be considered effectively anonymized and thereby free from regulatory restrictions under the GDPR.
Ultimately, our goal is to better understand a concept foundational to the GDPR, enabling a rigorous mathematical examination of whether certain classes of techniques (e.g., $k$-anonymity, differential privacy, pseudonymization) provide an important legal protection.

We are not the first to study this issue.
``Opinion on Anonymisation Techniques''~\cite{wp-anonymisation} provides guidance about the use of various privacy technologies---including $k$-anonymity and differential privacy---as anonymization techniques. It's analysis is centered on asking whether each technology effectively mitigates three risks: ``singling out, linkability, and inference.''
For instance, \cite{wp-anonymisation} concludes that with $k$-anonymity singling is no longer a risk whereas with differential privacy it ``may not'' be a risk.
Though similar in purpose to our work, its technical analyses are informal and coarse.
Reconsidering these questions with mathematical rigor, we encourage revisiting the conclusions in~\cite{wp-anonymisation}.

\subsection{Singling out in the GDPR}
\label{sec:soGDPR}
We begin with the text of the GDPR. It consists of {\em articles} detailing the obligations placed on processors of personal data as well as  {\em recitals} containing explanatory remarks. Article~1 of the regulation delineates its scope as ``lay[ing] down rules relating to the protection of natural persons with regard to the processing of personal data and rules relating to the free movement of personal data.'' The GDPR places no restrictions on the processing of non-personal data, even if this data is the result of \emph{anonymizing} personal data.\footnote{Recital 26 emphasizes this point: ``The principles of data protection should therefore not apply to anonymous information, namely information which does not relate to an identified or identifiable natural person or to personal data rendered anonymous in such a manner that the data subject is not or no longer identifiable.''}
Personal data is defined in Article~4 to mean ``any information relating to an identified or identifiable natural person; an identifiable natural person is one who can be identified, directly or indirectly.''
What it means for a person to be  ``identified, directly or indirectly'' is not elaborated in the articles of the GDPR.
Recital~26 sheds a little more light:
``To determine whether a natural person is identifiable account should be taken of all the means reasonably likely to be used, such as singling out, either by the controller or by another person to identify the natural person directly or indirectly.''
Singling out is one way to identify a person in data, and only data that does not allow singling out may be excepted from the regulation.\footnote{Interestingly, singling out is the only criterion for identifiability explicitly mentioned in the GDPR, the only occurrence the term being the quoted passage from Recital~26.}

For insight as to the regulation's meaning, we refer to two documents prepared by the Article~29 Data Protection Working Party, an advisory body set out by the EU Data Protection Directive.\footnote{Formally, Directive on the protection of individuals with regard to the processing of personal data and on the free movement of such data. 95/46/EC.}
``Opinion on the Concept of Personal Data''~\cite{wp-personal-data} elaborates on the meaning of ``identifiable, directly or indirectly.''
A person is identified ``within a group of persons [when] he or she is distinguished from all other members of the group.''
One way of distinguishing a person from a group is by specifying ``criteria which allows him to be recognized by narrowing down the group to which he belongs.''
If the group is narrowed down to an individual, that individual has been singled out.\footnotemark\
Looking ahead, we will call this \emph{isolating} an individual in the dataset and argue that not every instance of isolation should be considered a singling out attack.

\footnotetext{The notion of ``singling out'' is not defined in the Opinion on the Concept of Personal Data~\cite{wp-personal-data}. It is used in~\cite{wp-personal-data} four times, each consistent with the above interpretation. Our interpretation coincides with and was initially inspired by that of \cite{Diffix}, defining ``singling out as occurring when an analyst correctly makes a statement of the form `There is exactly one user that has these attributes.' "}

We highlight three additional insights that inform our work.
First, identification does not require a name or any other traditional identifier. For instance, singling out can be done with a ``small or large'' collection of seemingly innocuous traits (e.g., ``the man wearing a black suit'').
Indeed, this is what is meant by ``indirectly identifiable.''
An example of singling out in practice cited by \cite{wp-anonymisation} showed that four locations sufficed to uniquely identify 95\% of people in a pseudonymized dataset of time-stamped locations.
This is considered singling out even though no method of linking such location traces to individuals' names was identified.

Second, identifiable data may come in many forms, including microdata, aggregate statistics, news articles, encrypted data, video footage, and server logs.
What's important is not the form of the data, its whether the data permits an individual to be singled out.
We apply this same principle to the manner in which an individual is singled out within a dataset.
Most examples focus on specifying a collection of attributes (e.g., four time-stamped locations) that match a single person in the data. The collection of attributes corresponds to a \emph{predicate}: a function that assigns to each person in the dataset a value $0$ or $1$ (interpreted as $\mathsf{false}$ or $\mathsf{true}$ respectively).
We interpret the regulation as considering data to be personal if an individual can be distinguished within a dataset using any predicate, not only those that correspond to specifying collections of attributes.
Just as ``small or large'' collections of attributes may be used to single out, we allow these predicates to be simple or complex.

Third, whether or not a collection of attributes identifies a person is context-dependent.
``A very common family name will not be sufficient to identify someone - i.e. to single someone out - from the whole of a country's population, while it is likely to achieve identification of a pupil in a classroom.''
Both the \emph{prevalence} of the name and the \emph{size} of the group are important in the example, and will be important in our formalization.

\subsection{Our contributions}

\subsubsection{Defining security against predicate singling out}

In this work, we formalize and analyze \emph{predicate singling out}, a notion which is intended to partially model the GDPR's notion of singling out.
Following the discussion above, we begin with the idea that singling out an individual from a group involves specifying a predicate that uniquely distinguishes the individual, which we call \emph{isolation}.
Using this terminology, an intuitive interpretation of the GDPR's requirement is that to be considered secure against singling out, a function of the data must prevent isolation.
Trying to make this idea formal, we will see that it requires some refinement.

We restrict our attention to datasets $\db = (\row_1,\dots,\row_n)$ of size $n$, where each row $\row_i$ is sampled according to some underlying {probability distribution} $D$ over a universe $\univ$.
The dataset $\db$ is assumed to contain personal data corresponding to individuals, with at most one row per individual.
For example, $\db$ might consist of home listings, hospital records, internet browsing history, or any other personal information.
A mechanism $M$ takes $\db$ as input and outputs some data release $M(\db)$, be it a map of approximate addresses, aggregate statistics about disease, or pseudonymized internet histories.
We call $M$ an \emph{anonymization mechanism} because it purportedly anonymizes the personal data $\db$.

\newcommand{\y}{\mathbf{y}}
An {adversary} $\A$ attempts to output a {predicate} $p:\univ\rightarrow\zo$ that \emph{isolates} a row in $\db$, i.e., there exists $i$ such that $p(\row_i) = 1$ and $p(\row_j) = 0$ for all $j\neq i$.
We emphasize that it is rows in the original dataset $\db$ on which the predicate acts, not the output $\y$.
In part, this is a byproduct of our desire to make no assumptions on the form of $M$'s output. While it might make sense to apply a predicate to pseudonymized microdata, it is far from clear what it would mean for a synthetic dataset or for aggregate statistics.
Observe that this choice also rules out predicates $p$ that ``isolate'' rows by referring to their position in $\db$ (i.e., ``the seventh row'').

$M$ \emph{prevents isolation} if there doesn't exist an adversary $\A$ that isolates a row in $\db$ except with very small probability over the randomness of sampling $\db \gets D^n$, the mechanism $\y \gets M(\db)$, and the adversary $\A(\y)$.
Unfortunately, this is impossible to achieve by any mechanism $M$. To wit, there is a \emph{trivial adversary}---one that that doesn't look at $\y$ and denoted by $\T(\bot)$---that isolates a row with probability approximately $0.37$. The adversary simply outputs $p$ that matches a $1/n$ fraction of the distribution $D$.
For example, for a dataset of size $n=365$ random people selected at random from the United States population, $\T(\bot)$ may simply output $p= \mbox{(born on March 15th)}$. This predicate will isolate a row  with probability $${365 \choose 1}\cdot \frac{1}{365}\cdot\left(1-\frac{1}{365}\right)^{364} \approx 37\%.$$
Isolation is hence not necessarily indicative of a failure to protect against singling out, as $\T(\bot)$ would succeed with $\approx37\%$ probability (for any $n$) even if $M$ does not output anything at all. Furthermore, a trivial adversary need not knw the distribution $D$ to isolate with probabiilty $\approx 37\%$, as long as $D$ has sufficient min-entropy (Section~\ref{sec:bounding-base}).

A trivial adversary can give us a \emph{baseline} against which to measure isolation success.
But the baseline should not simply be 37\% chance of success.
Consider the earlier example of a dataset of 365 random Americans. What if an adversary output predicates like $p=(\mbox{born on March 15th} \wedge \mbox{vegan} \wedge \mbox{speaks Dutch} \wedge \mbox{concert pianist})$, and managed to isolate 10\% of the time?
Though 10\% is much less than 37\%, the predicate is extremely specific and unlikely to isolate a person by chance.
We formalize this intuition by considering the baseline risk of isolation as a function of the \emph{weight} of $p$, i.e., the chance that $p$ matches a random row sampled from the distribution $D$. The baseline for predicates of weight $1/n$ is 37\%, but the baseline for an extremely specific predicate may be much lower. The more specific the predicate, the closer the baseline gets to zero.
Our primary focus in this paper is on the regime of predicate weights where the baseline is negligible, corresponding to predicates with negligible weight.\footnote{For completeness, we also consider in Section~\ref{sec:definitions} predicates of weight $\omega(\log n/n)$, where the baseline is also negligible.} We get:

\smallskip \noindent {\bf Definition~\ref{def:security-against-singling-out}} (informal) An adversary \emph{predicate singles out} a row in $\db$ if it outputs a predicate that isolates a row with probability significantly higher than the baseline risk.
A mechanism $M$ is \emph{secure against predicate singling out} (\emph{PSO secure}) if no adversary can use its output to predicate single out.

\subsubsection{Analyzing security against predicate singling out}

Having formulated security against singling out, our next goal is to understand the guarantee it offers,  what mechanisms satisfy it, and how this concept relates to existing privacy concepts, including differential privacy and $k$-anonymity.

Two desirable properties of a privacy concept are robustness to \emph{post-processing} and to \emph{composition}.
The former requires that if a mechanism $M$ is deemed secure, then anything that can be computed using the outcome of $M$ should also be deemed secure. Hence, the outcome may be reused without creating additional privacy risk.
For instance, if a PSO-secure mechanism $M$ outputs microdata, then any statistics that can be computed from that microdata should also be PSO-secure.
It follows directly from the definition of PSO security that it is robust to post-processing.

We would like that the privacy risk of multiple data releases is not significantly greater than the accumulated risks of the individual releases. In this case, we say that the privacy concept composes.
We prove that PSO security does not compose, and give two examples of this failure. First, we show that releasing aggregate statistics is PSO-secure but fails to compose super-logarithmically many times.
A collection of $\omega(\log(n))$ counts may allow an adversary to isolate a row with probability arbitrarily close to one using a predicate with negligible weight (and negligible baseline).
Second, we construct less natural pair of mechanisms that individually are PSO-secure but together allow the recovery of a row in the dataset. The first mechanism extracts and outputs a secret encryption key from one part of $\db$. The second extracts the same key and uses it to encrypt the last row $\row_n\in\db$, outputting the corresponding ciphertext.
The mechanisms individually prevent predicate singling out, but together completely fail.

Next, we ask whether existing privacy concepts guarantee PSO security.
We already know that differential privacy is not necessary for PSO security as exact counts are PSO-secure but not differentially private.
However, differential privacy does provide PSO security.
The proof relies on the connection between differential privacy and statistical generalization guarantees~\cite{DFHPRR15, BNSSSU16}.
We show that predicate singling out implies a form of overfitting to the underlying dataset. If $M$ is differentially private it prevents this form of overfitting, and hence protects against predicate singling out.

Finally, we examine $k$-anonymity~\cite{SamaratiS98} and show that it does not prevent predicate singling out attacks.
Instead, it may enable an adversary to predicate single out with probability approximately 37\% using extremely low-weight predicates for which the baseline risk is negligible.
Briefly, the attack begins by observing that typical $k$-anonymous algorithms  ``almost'' predicate single out. They reveal predicates---usually, collections of attributes---that are satisfied by only $k$ rows in the dataset. In an effort to make the $k$-anonymized data as useful as possible, these predicates are as descriptive and specific as possible.
To predicate single out a row from the dataset of size $n$ using the $k$-anonymous output, it roughly suffices to predicate single out a row from any grouping of $k$ rows in the output.

\subsection{Implication for the GDPR}

Precisely formalizing predicate singling out attacks allows us to examine with mathematical rigor the extent to which specific algorithms and paradigms protect against them.
In particular, we show that $k$-anonymity fails to prevent predicate singling out, but that differential privacy prevents predicate singling out.
Our conclusions contrast with those of the Article~29 Working Party: they conclude that $k$-anonymity eliminates the risk of singling out while differential privacy ``may not''~\cite{wp-anonymisation}.
These disagreements may raise a doubt about whether our modeling indeed matches the regulators' intent.

Our goal in interpreting the text of the GDPR and related documents, and in defining predicate singling out, is to provide a precise mathematical formalism to capture some aspect of the concept of personal data (as elucidated in the regulation and in \cite{wp-personal-data}) and the associated concept of anonymization.
We want to render mathematically \emph{falsifiable} a legal claim that a given algorithmic technique anonymizes personal data by providing a \emph{necessary} condition for such anonymizers.

We argue that predicate singling out succeeds.
A number of modeling choices limit the scope of our definition,
but limiting the scope poses no issue. Specifically, (i) we only consider randomly sampled datasets; (ii) we only consider an attacker who has no additional knowledge of the dataset besides the output of a mechanism; (iii) we do not require that isolation be impossible, instead comparing to a baseline risk of isolation. A technique that purports to anonymize all personal data against all attackers must at least do so against randomly sampled data and against limited attackers.
And unless the idea of anonymization mechanisms is completely vacuous, one must compare against a baseline risk.

We must be careful not when narrowing our definition's scope, but when expanding it.
The most significant expansion\footnote{We discuss additional subtleties in Section~\ref{sec:knowledge-discussion}.} is our choice to parameterize the baseline risk by the weight of a predicate.
But this is a minimal expansion and only done to prevent a severe weakness.
Not doing so would mean that a mechanism that published the first row of the dataset 20\% of the time could be said to ``prevent singling out.''
Any meaningful instantiation of ``preventing singling out'' should rule out such mechanisms.
Ours is a natural way of doing so.

This does not mean that our modeling is the only one possible. As the starting point for the analysis is a description which does not use mathematical formalism, but is rather a (somewhat incomplete) description using natural language.
It is certainly plausible that alternative mathematical formalizations of singling out could be extracted from the very same text. We are looking forward to seeing such formalizations emerge.

Finally, one may still claim that the assessments made in~\cite{wp-anonymisation} should be taken as ground truth and that the Article~29 WP meant for any interpretation of singling out to be consistent with these assessments. That is, the protection provided by $k$-anonymity implicitly defines the meaning of singling out (partially or in full). We believe, however, that such a position would be hard to justify. To the best of our knowledge, the assessments made by the Article~29 WP were not substantiated by a mathematical analysis.
Furthermore, we caution against defining privacy implicitly as the guarantee provided by particular techniques; this approach is doomed to fail.
In particular, the choice of defining privacy as the result of applying practices such as suppression of directly identifying information has proved a problematic choice that unfortunately pervades current legal privacy standards.

\paragraph {Is predicate singling out a good privacy concept?}

A predicate singling out attack can be a stepping stone towards a greater harm, even in settings where isolation alone may not.
It may enable linking a person's record in the dataset to some external source of information~\cite{narayanan2008robust}, or targeting of individuals for differential treatment.
As such, it is meaningful as a mode of privacy failure, both in the GDPR context and otherwise.

And, while we believe that PSO security is relevant for the GDPR as a necessary property of techniques that anonymize personal data, we do not consider it a sufficiently protective privacy concept by itself. First, singling out is a specific mode of privacy failure.
It is not clear that ruling out this failure mode is sufficient for privacy (in particular, two other failure modes are mentioned in~\cite{wp-anonymisation}: linkage and inference).
Second, our definition considers a setting where the underlying data is chosen i.i.d.\ from some (unknown) underlying distribution, an assumption that is not true in many real-life contexts. PSO security may not prevent singling out in such contexts. Lastly, we believe that self-composition is an essential property of any reasonable privacy definition. However, as we show in Section~\ref{sec:failureToCompose}, security against singling out does not self compose.


\section{Preliminaries}
\label{sec:preliminaries}

\paragraph{Notation.}
A dataset $\db=(\row_1,\ldots,\row_n)$ consists of $n$ elements taken from the data $\univ = \zo^d$.
We consider datasets where each entry $\row_i$ is independently sampled from a fixed probability distribution $D\in\Delta(\univ)$ over $\univ$. We denote by $U_d$ a uniform random variable over $\zo^d$.

For the purposes of asymptotic analyses, we will use the number of rows $n\in \N$ in a dataset as the complexity parameter.
Furthermore, the parameter $d = d(n)$ is a function of $n$, but we typically omit the dependence. \footnotemark~
A function $f(n)$ is negligible, denoted $f(n) = \negl(n)$ if $f(n) = n^{-\omega(1)}$.

	\footnotetext{More formally, we can consider an ensemble of data domains $\cX= \{\univ_n = \zo^{d(n)}\}_{n\in \N}$ and an ensemble of distributions $\cD = \{D_n\}_{n\in\N}$, where $D_n\in\Delta(\univ_n)$.}

A mechanism $M$ is a Turing Machine that takes as input a dataset $\db \in \univ^n$. A mechanism $M$ may be randomized and/or interactive.
A predicate is a binary-valued function $p:\univ\rightarrow\zo$. We define $\weight_D(p) \triangleq \E_{\row \sim D}[p(\row)]$. For a dataset $\db\in \univ^n$ we define  $p(\db) \triangleq \frac{1}{n}\sum_{i=1}^n p(\row_i)$.
We occasionally use indicator notation $\indic{}()$ to define a predicate: for example, $p(x) = \indic{x\in A}$ equals 1 if $x\in A$ and 0 otherwise.

\subsection{Preliminaries from randomness extraction}

\begin{definition}[Min-entropy, average min entropy \cite{fuzzy}]Let $Y_1,Y_2$ be two random variables.
The \emph{min-entropy} of a $Y_1$ is  $${H}_\infty(Y_1) = -\log\left(\max_{y} \Pr[Y_1 = y] \right).$$
The \emph{average min-entropy}\footnotemark of $Y_1$ given $Y_2$ is $$\widetilde{H}_\infty(Y_1 \mid Y_2) = -\log\left(\E_{Y_2} \biggl[\max_{y} \Pr[Y_1 = y \mid Y_2]\biggr]\right).$$
\end{definition}
\footnotetext{In \cite{smith2009foundations} this same quantity is called conditional min-entropy and denoted $H_\infty$.}

\begin{fact}
		For all $Y_1$ and $Y_2$: ${H}_\infty(Y_1) \ge \widetilde{H}_\infty(Y_1 \mid Y_2) \ge {H}_\infty(Y_1) - \log (|\supp(Y_2)|)$, where $\supp(Y_2)$ is the support of $Y_2$.
\end{fact}

\begin{definition}[2-universal hash functions]
$H = \{h:\zo^d \to \zo^m\}$ is a 2-universal family of hash functions if $\Pr_{h \sim H} [h(\row) = h(\row')] = 2^{-m}$ for all $\row,\row'\in \zo^d$ where the probability is over the selection of $h$ uniformly at random from $H$.
\end{definition}

As an example, for $a,b\in\zo^d$ let $h_{a,b}(x)$ be the function that returns the first $m$ bits of $ax + b$ where the arithmetic is in the field $GF(2^d)$. Then $H=\{h_{a,b}:a,b\in\zo^d\}$ is 2-universal.

\begin{definition}[Statistical distance] The statistical distance of random variables  $Y_1,Y_2$ with support $\zo^d$ is $\SD(Y_1,Y_2) = \frac{1}{2} \sum_{y\in \zo^d} \bigl|\Pr[Y_1=y] - \Pr[Y_2=y]\bigr|$. If $\SD(Y_1,Y_2) < \alpha$ we say that $Y_1$ and $Y_2$ are $\alpha$-close.
\end{definition}

\begin{lemma}[Generalized Leftover Hash Lemma~\cite{fuzzy}]
	\label{LHL}
Let $\lambda\in \N$, $\alpha > 0$, $Y_1$ a random variable over $\zo^d$, and $Y_2$ a random variable. Let $H = \{h:\zo^d \to \zo^m\}$ be a 2-universal family of hash functions where $m \le \lambda-2\log(1/\alpha^2)+2$. For every $Y_1$ and $Y_2$ with $\widetilde{H}_\infty(Y_1 \mid Y_2)\ge\lambda$,
$(h,h(Y_1),Y_2)$ is $\alpha^2$-close to  $(h, U_m, Y_2)$ in total variation distance, where $h\in_R H$ is a uniformly random function from the hash family and $U_m$ is uniform over $\zo^m$.
\end{lemma}

\begin{corollary}
	\label{cor:sd-fact}
	$h(Y)$ is ${\alpha}$-close to uniform with probability at least $1-\alpha$ over $h\in_R H$.
\end{corollary}
\begin{proof}
Let $H_{> \alpha} = \{h\in H : h(Y)~\mbox{is not}~\alpha~\mbox{close to uniform}\}$. We have
$\alpha^2 \geq \Delta ((h,h(Y)), \mbox{unif}) > \Pr[h \in H_{> \alpha}] \cdot \alpha$. Hence $\Pr[h \in H_{> \alpha}] < \alpha$.
\end{proof}


\section{Security against predicate singling out (PSO security)}
\label{sec:definitions}

We consider a setting in which a data controller has in its possession a dataset $\db = (\row_1, \ldots, \row_n)$ consisting of $n$ rows  sampled i.i.d.\ from a distribution $D\in\Delta(\univ)$.
The data controller publishes the output of an \emph{anonymization mechanism} $\mech$ applied to the dataset $\db$.
A predicate singling out (PSO) adversary $A$ is a non-uniform Turing machine with access to the mechanism $M(\db)$ and produces a predicate $p:\univ\to\zo$.\footnotemark~
We abuse notation and write $\adv(M(\db))$, regardless whether $M$ is an interactive or non-interactive mechanism.
For now, we assume all adversaries have complete knowledge of $D$ and are computationally unbounded; we reexamine these choices in Section~\ref{sec:knowledge-discussion} below.

\footnotetext{
As is typical in cryptography, strengthening the adversary to be non-uniform (including possibly having full knowledge of the distribution $D$) yields stronger security definition. See Section~\ref{sec:knowledge-discussion} for further discussion.
} 

Intuitively, the adversary's goal is to output predicate $p$ that isolates a row in $\db$, where we associate the Article~29 WP Opinion on Anonymisation Techniques notion of ``isolat[ing] some or all records which identify an individual in [a] dataset'' with the production of a description that matches exactly one row in the dataset. Mathematically, the description would be in form of a predicate mapping data universe elements into $\{0,1\}$.

\begin{definition}[Row isolation]
A predicate $p$ \emph{isolates a row in $\db$} if there exists a unique $\row \in \db$ such that $p(\row) = 1$. I.e., if $p(\db) = 1/n$. We denote this event $\iso{p}{\db}$.
\end{definition}

It is tempting to require that a mechanism $M$ only allow a negligible probability of isolating a row, but this intuition is problematic.
An adversary that does not have access to $M$---a \emph{trivial} adversary---can output a predicate $p$ with $\weight_D(p)\approx 1/n$ and hence isolate a row in $\db$ with probability ${n \choose 1} \cdot \weight_D(p)\cdot (1-\weight_D(p))^{n-1}\approx e^{-1}\approx 37\%$.
In Section~\ref{sec:bounding-base} we will see that in many cases the trivial adversary need not know the distribution to produce such a predicate.

\medskip

Instead of considering the absolute probability that an adversary outputs a predicate that isolates a row, we consider the increase in probability relative to a \emph{baseline risk}: the probability of isolation by a trivial adversary.

\begin{definition}[Trivial Adversary]
	A predicate singling out adversary $\T$ is \emph{trivial} if the distribution over outputs of $\T$ is independent of $M(\db)$. That is $\T(M(\db)) = \T(\bot)$.
\end{definition}

An unrestricted trivial adversary can isolate a row with probability about $1/e$.
Towards a more expressive notion of the baseline risk, we restrict adversaries to output a
predicate from a particular class of {\em admissible} predicates $P \subseteq \{p:\univ\to \zo\}$, i.e., a subset of predicates on $\univ$.\footnote{More formally, we restrict the adversary to an ensemble of admissible predicates $\cP = \{P_n\}_{n\in\N}$, where $P_n \subseteq \{p:\univ_n\to \zo\}$, a subset of predicates on $\univ_n$.}

\begin{definition}[Adversarial success probability]
	\label{def:success}
	Let $D$ be a distribution over $\univ$.
	For mechanism $M$, an adversary $\A$, a set of admissible predicates $P$, and $n\in \bbN$, let
	$$\Succ_{P}^{\adv, M}(n,D) \triangleq \Pr_{\substack{\db\gets D^n \\ p\gets \adv(M(\db))}}[\iso{p}{\db} \land p\in P].$$
\end{definition}

\begin{definition}[Baseline]
  \label{def:baseline}
	For $n\in\N$ and set of admissible predicates $P$,
	$$
	\base_D(n,P) \triangleq \sup_{\mbox{\scriptsize Trivial }\T} ~ \Succ_{P}^{\T,\bot}(n,D)
	$$
\end{definition}

We typically omit the parameter $D$ when the distribution is clear from context. In this work, we focus on two classes of admissible predicates parameterized by the \emph{weight} of the predicate $p$.

\begin{definition}[Predicate families $\low$ and $\high$]
	For $0\le \w_\ell(n) \le 1/n \le \w_h(n) \le 1$ we define the predicate families
$$ \low  = \{p: \weight_D(p) \leq \w_\ell(n)\} \quad\mbox{and}\quad \high = \{p: \weight_D(p) \geq \w_h(n)\}.$$
\end{definition}
We will consider the success probability of adversaries restricted to these admissible predicates and will denote them $\Succ_{\le \w_\ell}^{\adv,M}$ and $\Succ_{\ge \w_h}^{\adv,M}$ as shown in Figure~\ref{fig:succ-low}.

\begin{figure}[h]
    \centering
    \medskip
    \includegraphics[width=\textwidth]{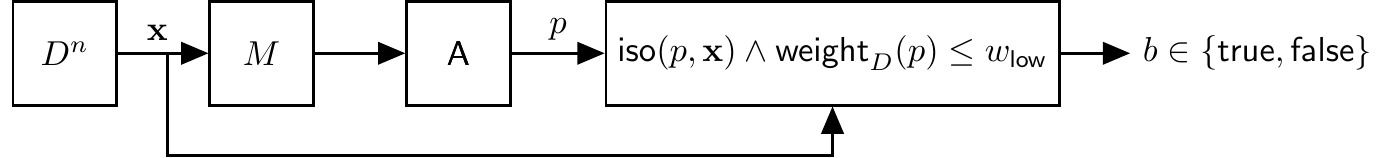}
		\caption{\label{fig:succ-low} $\mathsf{Succ}_{\le w_{\mathsf{low}}}^{\mathsf{A},M}(n,D) = \Pr_{D,M,\mathsf{A}}[b = \mathsf{true}].$}

\end{figure}

\subsection{Security against predicate singling out}

We now have the tools for presenting our definition of security against singling out. We require that no adversary should have significantly higher probability of isolating a row than that of a trivial adversary, conditioned on both outputting predicates from the same class of admissible predicates.

\begin{definition}[Security against predicate singling out\label{def:security-against-singling-out}]
	For $\epsilon(n) > 0$, $\delta(n) > 0$, $0 \le \w_\low(n) \le 1/n \le \w_\high(n) \le 1$, we say a mechanism $M$ is \emph{$(\epsilon, \delta, \w_\low,\w_\high)$ secure against predicate singling out}
  ($(\epsilon, \delta, \w_\low,\w_\high)$\emph{-PSO secure})
  if for all $\adv$, $D$, $n$, $w_\ell \le w_\low$, and $w_h \ge w_\high$:
	\begin{align}
	\Succ_{\le\w_\ell}^{\adv,M}(n,D) &\le e^{\epsilon(n)} \cdot \base_D(n,\low) + \delta(n),  \nonumber\\
	\Succ_{\ge\w_h}^{\adv,M}(n,D)  &\le e^{\epsilon(n)} \cdot  \base_D(n,\high) + \delta(n). \label{eq:succ_w_high}
	\end{align}
	We often omit explicit reference to the parameter $n$ for $\epsilon$, $\delta$, $\w_\low$, and $\w_\high$.

We say a mechanism is \emph{secure against predicate singling out} (\emph{PSO secure}) if for all $w_\low = \negl(n)$, $\w_\high = \omega(\frac{\log n }{n})$ there exists $\delta = \negl(n)$ such that $M$ is $(0,\delta,\w_\low,\w_\high)$-PSO secure.
\end{definition}

The definition is strengthened as $\epsilon$ and $\delta$ get smaller, and as $\w_\low$ and $\w_\high$ get closer to $1/n$. As shown below, when $\w_\low = \negl(n)$ the baseline is negligible. This is probably the most important regime of Definition~\ref{def:security-against-singling-out} as such predicates are likely to not only isolate a row in the dataset but also an individual in the entire population. The baseline is also negligible when  $\w_\high = \omega(\log n / n)$. It is not clear to the authors how beneficial finding a predicate in this regime may be to an attacker. The reader may decide to ignore Equation~\ref{eq:succ_w_high} in Definition~\ref{def:security-against-singling-out} (as is depicted in Figure~\ref{fig:succ-low}). We include the high weight regime in our analysis so as not to overlook potential singling out risks which rely on high weight predicates.

We also define a strong notion of predicate singling out, where an adversary can simultaneously isolate all rows of a dataset.
\begin{definition}[Fully Predicate Singling Out]
	\label{def:blatant}
	An adversary $\A$ \emph{fully singles out} against a mechanism $M$ and distribution $D$ if (with high probability) it outputs a collection of $n$ negligible-weight predicates $p_i$, each of which isolates a different row of the input dataset $\db$. More formally, if
	\begin{equation}
		\Pr_{\substack{\db\gets D^n \\ (p_1,\dots,p_n)\gets \A(M(\db))}}[\forall p_i, p_j: \iso{p_i}{\db} \land \weight_D(p_i) = \negl(n) \land (p_i \land p_j)(\db) = 0 ] > 1 - \negl(n)
	\end{equation}
\end{definition}

\paragraph{Examples.}
On input $(\row_1,\dots,\row_n)$ the mechanism $M_f$ outputs $(f(\row_1),\dots,f(\row_n))$ for some possibly randomized function $f$.
Whether $M_f$ prevents predicate singling out depends on $f$.
On one extreme, if $f(\row) = \row$ and $|\univ|\gg n$, then $M_f$ provides no protection. On the other extreme, if $f(\row)$ is completely random, $M_f(\db)$ contains no information about $\db$ and provides no benefit to the adversary. More formally, for all $\db$ the output of $M_f(\db)$ is uniform; this allows us to construct a trivial adversary $\T$ that perfectly simulates any adversary $\A$.\footnotemark

\footnotetext{Uniformity without conditioning on $\db$ may not be enough. For example, if the data itself is uniform, then the output of the identity function is also uniform. See also footnote~\ref{foot:uniformity-1}.}

If $f$ is invertible, then it offers no more protection than the identity function. However, $f$ being many-to-one does not give an assurance. For instance, suppose the data is uniform over $\zo^n$ and $f:\zo^n \to \zo^{n/2}$ outputs the last $n/2$ bits of an input $\row$.
$M_f$ is \emph{not} secure.
Indeed, it allows fully predicate singling out. For any $y_i = f(\row_i)$ in the output, the adversary can output the predicate $p_i:(\row) \mapsto \indic{f(\row) = y_i}$.
$\Pr[\iso{p_i}{\db}] = 1-\negl(n)$ and $\weight_{U_n}(p_i) = 2^{-n/2}= \negl(n)$.

\subsection{Bounding the baseline}
\label{sec:bounding-base}
In this section, we characterize the baseline over intervals in terms of a simple function $B(n,w)$.
For $n\ge2$ and a predicate $p$ of weight $w$, the probability over $\db\sim D^n$ that $p$ isolates a row in $\db$ is
\begin{equation*}
	B(n,w) \triangleq n \cdot w \cdot (1-w)^{n-1}
\end{equation*}
$B(n,w)$ is maximized at $w = 1/n$ and strictly decreases moving away from the maximum.
It is helpful to recall that $(1-1/n)^n \approx e^{-1}$ even for relatively small values of $n$.
$(1-1/n)^{n-1}$ also approaches $e^{-1}$ as $n\to\infty$, and does so from above.

As made formal in Claim~\ref{claim:baseline-exact} (proof in Appendix~\ref{app:definitions}), a trivial adversary maximizes its success of isolating a row by outputting a predicate $p$ with $\weight_D(p)$ as close as possible to $1/n$ (the weight that maximizes $B(n,w)$).
The set of possible values for $\weight_D(p)$ depends not only on $\w_\low$ and $\w_\high$, but also on the distribution.
We say that a weight $\w \in [0,1]$ is \emph{realizable} under distribution $D$ if there exists $p$ such that $\weight_D(p) = \w$. The baseline is characterized by $B(n,w)$.
\begin{claim}
	\label{claim:baseline-exact}
	For every $n > 0$, $w_\low$, $w_\high$ and $D$,
$$ \base_D(n,\lowset_n) = B(n,\w^*_\low(n))
\quad\mbox{and}\quad
\base_D(n,\highset_n) = B(n,\w^*_\high(n)),$$
where
$$\w^*_\low(n) = \sup \{w \le \w_\low(n) : \mbox{ realizable}\}
\quad\mbox{and}\quad
\w^*_\high(n) = \inf \{w \ge \w_\high(n) : \mbox{ realizable}\}. $$
\end{claim}

Because $B(n,w)$ increases as $w$ approaches $1/n$, the baseline has a simple upper-bound.
\begin{corollary}
	\label{clm:baselineUpperBound}
		For every $w_\low$, $w_\high$, $n\in\N$ and distribution $D$,
$$\base_D(n,\lowset_n) \le B(n,\w_\low(n)) \quad
\mbox{and}\quad
\base_D(n,\highset_n) \le B(n,\w_\high(n)).$$
\end{corollary}

The dependence on the realizability of weights under $D$ makes the exact  baseline unwieldy. For example, the difference between the true baseline and the upper bound can be as large as $1/e$.
Thankfully, the $B(n,w)$ upper bound is nearly tight when the underlying distribution has moderate min-entropy. Moreover, the corresponding lower bound is achievable by an efficient uniform trivial adversary who is oblivious of the distribution (see Section~\ref{sec:knowledge-discussion}).

\begin{claim}[Baseline Lower Bound]
	\label{baseline-lower-bound}
	Let $c>0$ and $0\le \w_\low(n) \le 1/n \le \w_\high(n) \le 1$.
	If $D$ has min-entropy at least $\lambda> 5(c+\log n + 2)$, then
$\base_D(n,\lowset_n) \ge B(n,\w_\low(n)) - 2^{-c}$ and $\base_D(n,\highset_n) \ge B(n,\w_\high(n)) - 2^{-c}$.
\end{claim}

Informally, the assumption that $D$ has min-entropy $\lambda$  implies two useful facts. First, the set of realizable weights is dense: for any $w$, there exists a realizable $w'$ such that $|w-w'|$ is small. Second, the Leftover Hash Lemma allows us to construct an efficient uniform adversary who can find a predicate with weight $w'$ without any knowledge of the distribution.
The following lemma (proved in Appendix~\ref{app:definitions}) captures these properties:

\begin{lemma}
	\label{lemma:lhl-1}
	For $m\in\N$ and a set $\univ$, let $H = \{h:\univ\to\zo^m\}$ be 2-universal family of hash functions. 
	For any $w \ge 2^{-(m-1)}$ (respectively, $w \le 1-2^{-(m-1)}$), there exists a collection of predicates $P_H = \{p_h\}_{h\in H}$ such that for all distributions $D$ over $\univ$ with min-entropy at least $\lambda = 5m$,
	$\weight_D(p_h) \in [\w - 3\cdot2^{-m}, \w]$
	with probability at least $1-2^{-m}$ over $h\in_R H$.
	(respectively, $\weight_D(p_h) \in [\w, \w + 3\cdot2^{-m}]$).
\end{lemma}

\begin{proof}[Proof of Claim~\ref{baseline-lower-bound}]
	We prove the claim for $\w_\low(n)$; the proof for $\w_\high(n)$ is analogous.
  Let $m\ge c+\log n +2$. Either $w_\low(n) \le 2^{-(c+\log n)}$ or $\w_\low(n) \ge 2^{-(m-1)}$.
  If $w_\low(n) \le 2^{-(c+\log n)}$, then $B(n,w_\low) \le nw_\low \le 2^{-c}$, making the claim trivial.
  It remains to consider $\w_\low(n) \ge 2^{-(m-1)}$.

	Let $P_H$ be the family of predicates from Lemma~\ref{lemma:lhl-1} and $\T_H$ be a trivial adversary that outputs a random $p_h \in_R P_H$.
	Recall that for any predicate $p$, $\Pr_{\db \sim D^n}[\iso{p}{\db}]=B(n, \weight_D(p))$.
	By Lemma~\ref{lemma:lhl-1}
	\begin{align*}
		\Succ_\low^{\T_H,\bot}(n)
		&\ge  \Pr_{\db \sim D^n, h\in_R H}[\iso{p_h}{\db} \land \weight_D(p_h) \in W_\low] \\
		&= \Pr_{h\in_R H}[\weight_D(p_h) \in W_\low]\cdot \Pr_{\db \sim D^n, h\in_R H}[\iso{p_h}{\db} \mid \weight_D(p_h) \in W_\low] \\
		&\ge (1 - 2^{-m})\cdot B(n,3\cdot 2^{-m})
	\end{align*}
	Observing that 	$\left|\dv{B}{w}(w)\right| \le \dv{B}{w}(0) =n$,
	$\Succ_\low^{\T_H,\bot}(n) \ge B(n,\w_\low(n)) - 3\cdot2^{-m}n - 2^{-m} \ge B(n,\w_\low(n)) - 2^{-(m-\log n - 2)} \ge B(n,\w_\low(n)) - 2^{-c}$.
\end{proof}

\begin{remark}
	\label{remark:baseline-derandomized}
	The proof of Claim~\ref{baseline-lower-bound} requires only that is possible to sample a predicate such that $\weight_D(p_h) \in W_\low$. If we switch the order of quantifiers in the claim by allowing the trivial adversary to depend on the distribution $D$, then the proof (and thus the trivial adversary) can be derandomized. Indeed, for any $D$ with sufficient min-entropy, there are many $p_h$ that can be used. This observation is used in the proof of Theorem~\ref{thm:count-composition-uniform}.
\end{remark}

\subsection{Reflections on modelling assumption}
\label{sec:knowledge-discussion}

In many ways, Definition~\ref{def:security-against-singling-out} requires a very high level of protection, similar to what is standard in the foundations of cryptography. The definition requires a mechanism to provide security for all distributions $D$ and against non-uniform, computationally unbounded adversaries.\footnote{\label{foot:knowledge} It is reasonable to limit the adversary in Definition~\ref{def:security-against-singling-out} to polynomial time. If we restricted our attention to distributions with moderate min-entropy, our results would remain qualitatively the same: our trivial adversaries and lower bounds are all based on efficient and uniform algorithms; our upper bounds are against unbounded adversaries. Relatedly, restricting to min-entropy distributions would allow us to switch the order of quantifiers of $D$ and $\T$ in the definition of the baseline without affecting our qualitative results.}
The main weakness in the required protection is that it considers only data that is i.i.d., whereas real-life data cannot generally be modeled as i.i.d.

Any mechanism that purports to be a universal anonymizer of data under the GDPR---by transforming personal data into non-personal data---must prevent singling out.
Our definition is intended to capture a necessary condition for a mechanism to be considered as rendering data sufficiently anonymized under the GDPR.
Any mechanism that prevents singling out in all cases must prevent it in the special case that the data is i.i.d. from a distribution $D$.
We view a failure to provide security against predicate singling out (Definition~\ref{def:security-against-singling-out}) or is fully predicate singling out (Definition~\ref{def:blatant}) as strong evidence that a mechanism does not provide security against singling out; hence, it does not protect from identification, as per the analysis in Section~\ref{sec:soGDPR}.

On the other hand, satisfying Definition~\ref{def:security-against-singling-out} is {\em not sufficient} for arguing that a mechanism renders data sufficiently anonymized under the GDPR. Singling out is only one of the many ``means reasonably likely to be used'' to identify a person in a data release.\footnote{Article~29 Working Party Opinion on Anonymisation techniques~\cite{wp-anonymisation} enumerates  three criterions for identification: singling out, linkage, and inference.}
Furthermore, the definition considers only i.i.d.\ data; satisfying it may not even be sufficient to conclude that a mechanism prevents singling out in all relevant circumstances.


\section{Properties of PSO security}
\label{sec:composition}

Two desirable properties of privacy concepts are that (i) immunity to post-processing, i.e., further processing of the outcome of a mechanism, without access to the data, should not increase privacy risks, and (ii) closure under composition, i.e., a combination of two or more mechanisms which satisfy the requirements of the privacy concept is a mechanism that also satisfies the requirements (potentially, with worse parameters). Differential privacy is an example of a privacy concept that is immune to post-processing and is closed under composition.

In this section we prove that PSO security withstands post processing but not composition.
We give two demonstrations for the latter. In the first we consider mechanisms which count the number of dataset rows satisfying a property. We show that releasing a count satisfies Definifion~\ref{def:security-against-singling-out}. However, there exists a collection of $\omega(\log(n))$ counts which allows an adversary to isolate a row with probability arbitrarily close to one using a predicate with negligible weight. For the second demonstration, we construct a (less natural) pair of mechanisms that individually satisfy Definifion~\ref{def:security-against-singling-out} but together allow the recovery of a row in the dataset. This latter construction borrows ideas from~\cite{NSSSU18}.
An immediate conclusion is that PSO security is distinct from differential privacy. More importantly, not being closed under composition is a significant weakness of the notion of PSO security.
Our constructions rely on very simple mechanisms that would likely be deemed secure against singling out under other formulations of the concept. It may well be that non-closure under composition is inherent for singling out.

From a legal or policy point of view, we believe that a privacy concept which is not closed under composition (or not immune to post-processing) should not be accepted as sufficient. Pragmatically, the fact that PSO security is not closed under composition suggests that this concept can be used for {\em disqualifying} privacy technology (if they are not PSO secure) but also that this concept must be combined with other requirements if it used for approving technology.

\subsection{Post Processing}

For any non-interactive mechanism $M$, let $F$ be a (possibly non-uniform) algorithm taking inputs of the form $M(\db)$. Let $F\circ M$ be the mechanism that on input $\db$ returns $F(M(\db))$.

\begin{lemma}[Postprocessing]
\label{lemma:postprocessing}
If $M$ is $(\eps, \delta, \w_\low,\w_\high)$-PSO secure, then $F\circ M$ is too.
\end{lemma}

\begin{proof}
  We show something stronger: for all $M$, $F$, $\A$ there exists $\A_F$ such that for all $n$, $P$, $D$:
  $\Succ_P^{\A_F,M}(n) = \Succ_P^{\A,F\circ M}(n)$.
	On input $M(\db)$, $\A_F$ simulates $\A$ on input $F(M(\db))$ and returns the resulting predicate $p$. The distribution of $\A_F$'s output with mechanism $M$ is identical to that of $\A$ with mechanism $F\circ M$, proving the lemma.
\end{proof}

\noindent
The definition and proof above extend to the case where the mechanism $M$ is interactive.

\subsection{Example PSO-secure mechanisms}
\label{sec:counting-leakage}

This section presents two PSO-secure mechanisms.
These examples are useful for developing intuition for the PSO security notion. Additionally, they are the foundation for the examples of self-composition failures in the next section.

\subsubsection{Counting Mechanism}
For any predicate $q:\univ \to \zo$, we define the corresponding Counting Mechanism:
\begin{mechanism}[H]
		\caption{Counting Mechanism $M_{\#q}$}
		\SetKwInOut{Input}{input}
		\Input{$\db$}
		\BlankLine
		return $|\{1\leq i \leq n  : q(\row_i) = 1\}|$
\end{mechanism}
\noindent
For example, consider the least-significant bit predicate $\lsb$, that takes as input a string $\row \in \zo^*$ and outputs $\row[1]$. The corresponding Counting Mechanism $M_{\#\lsb}$ returns the sum of the first column of $\db$.

The security of the Counting Mechanism is a corollary of the following proposition.

\begin{proposition}
\label{lemma:mech-small-codomain}
	For all $\A$, $P$, $M:\univ^n\mapsto Y$: $\Succ_P^{\A,M}(n) \le |Y|\cdot \base(n, P),$ where $Y$ is the codomain of $M$.
\end{proposition}

\begin{proof}
We define a \emph{trivial adversary} $\T$ such that for all $\A$, $\Succ^{\T,\bot}_P(n) \ge \frac{1}{|Y|}\cdot\Succ^{A,M}_P(n).$
The proposition follows by definition of $\base(n,P)$.
$\T$ samples a random $y\in_R Y$ and returns $p \gets \A(y)$.
\begin{equation*}
\Succ_P^{\T,\bot}(n)
= \Pr_{\substack{\db \gets D^n \\ y\in_R Y \\ p\gets \A(y)}}[\iso{p}{\db} \land p\in P]
\ge \frac{\Succ^{A,M}_P}{|Y|}
\end{equation*}
The inequality follows from the fact that for all datasets $\db$, there exists $y^* = y^*(\db) \in Y$ such that
\begin{equation*}
\Pr_{p\gets \A(y^*)}[\iso{p}{\db} \land p \in P] \ge \Pr_{p\gets \A(M(\db))}[\iso{p}{\db} \land p \in P],
\end{equation*}
and that for all $\db$,
$
\Pr_{y\in_R Y}[y = y^*] \ge \frac{1}{|Y|}.
$
\end{proof}

\begin{corollary}
	\label{thm:count-bit}
	$M_{\#q}$ PSO secure.
\end{corollary}
As exact counts are not differentially private, this corollary demonstrates that differential privacy is not necessary for PSO security.

\subsubsection{Predicate Mechanism}

For any predicate $q:\univ \to \zo$, we define the corresponding Predicate Mechanism:
\begin{mechanism}[H]
		\caption{Predicate Mechanism $M_q$}
		\SetKwInOut{Input}{input}
		\Input{$\db$}
		\BlankLine
		return $(q(x_1), q(x_2),\dots,q(x_n))$
\end{mechanism}
\noindent

\begin{theorem}
	\label{thm:mech-one-bit}
	$M_q$ is PSO secure.
\end{theorem}

We prove the security of $M_q$ by showing that its output is ``no more helpful'' to the PSO adversary than the counts returned by $M_{\#q}$.

\begin{proposition}[Permutation Proposition]
	\label{lemma:permutation}
	For a permutation $\sigma:[n]\to[n]$ of $n$ elements and a dataset $\db=(\row_1,\dots,\row_n)$, define $\sigma(\db) = (\row_{\sigma(1)},\row_{\sigma(2)},\dots,\row_{\sigma(n)})$.
For any mechanism $M$, let $M\circ \sigma$ be the mechanism that on input $\db$ returns $M(\sigma(\db))$.
	For all $\A$, $P$, $D$, and $\sigma$: $\Succ_P^{\A,M}(n) = \Succ_P^{\A, M\circ\sigma}(n).$
\end{proposition}

\begin{proof}
	For all $\sigma$, the distributions $D^n$ and $\sigma(D^n)$ are identical.
	For all $p$ and $\db$, $\iso{p}{\db}$ if and only if $\iso{p}{\sigma(\db)}$. Using these two observations:
	\begin{align*}
		\Succ_P^{\A,M}(n)
	&= \Pr_{\substack{\db \gets D^n \\ p\gets\adv(M(\db))}}[\iso{p}{\db} \land p \in P] \\
	&= \Pr_{\substack{\db \gets D^n \\ p\gets\adv(M\circ\sigma(\db))}}[\iso{p}{\sigma(\db)} \land p \in P] \\
	&= \Pr_{\substack{\db \gets D^n \\ p\gets\adv(M\circ\sigma(\db))}}[\iso{p}{\db} \land p \in P] \\
	&=\Succ_P^{\adv,M\circ\sigma}(n)
  \tag*{\qedhere}
	\end{align*}
\end{proof}

\begin{proof}[Proof of Theorem~\ref{thm:mech-one-bit}]
	Consider $M_1$ that on input $\db$ samples a random permutation $\sigma$ and returns $M_q\circ \sigma(\db)$. By the Permutation Proposition, $\Succ_P^{\A,M_1}(n) = \Succ_P^{\A,M_q}(n)$.
	Next, consider the randomized algorithm $F$ that on input $m\in [n]$ outputs a uniformly random bitstring $y \in \zo^n$ of Hamming weight $m$. By Postprocessing and the security of $M_{\#q}$, the mechanism $M_2 = F\circ M_{\#q}$ is PSO secure.

	$M_1$ and $M_2$ are the same mechanism: on every input $\db$, the output distributions are identical. Therefore $M_q$ is PSO secure.
\end{proof}

\subsection{Failure to Compose}
\label{sec:failureToCompose}

\subsubsection{Failure to compose $\omega(\log n)$ times}
\label{sec:failcomposelog}

The security of a single count (Corollary~\ref{thm:count-bit}) easily extends to $O(\log n)$-many counts (even adaptively chosen), as the size of the codomain grows polynomially. However, our next theorem states that a fixed set of $\omega(\log(n))$ counts suffices to predicate single out with probability close to $e^{-1}$ (which can be amplified to $1-\negl(n)$).

\begin{theorem}
	\label{thm:count-composition-uniform}
  For a collection of predicates $Q = (q_0,\dots,q_m)$, let $M_{\#Q}(\db) \triangleq (M_{\#q_0}(\db),\dots,M_{\#q_m}(\db))$.
  Let $\univ = \zo^m$ and $D = U_m$ the uniform distribution over $\univ$.
  There exists $Q$ and an adversary $\A$ such that
  $$\Succ^{\A,M_{\#Q}}_{\le2^{-m}}(n) \ge B(n,1/n) - \negl(n). $$
\end{theorem}
\noindent
Choosing $m = \omega(\log(n))$ yields $2^{-m} = \negl(n)$.

\begin{proof}
Let $q_0$ be any predicate such that $\weight_{U_m}(q_0) \le 1/n$ such that $\Pr_{\db \gets U_m^n}[\iso{q_0}{\db}] \ge B(n,1/n) - \negl(n).$
For instance, $q_0(\row) = 1$ iff $\row < 2^m/n$ (where in the last inequality we treat $\row$ as a number written in binary).\footnote{Or use Claim~\ref{baseline-lower-bound} with $\w_\low(n) = 1/n$, and Remark~\ref{remark:baseline-derandomized}.}

For $i \in \{1,\dots,m\}$, define the predicate $q_i(\row) \triangleq (q_0(\row) \land \row[i])$, and let $y_i = M_{\#q_i}(\db)$.
Consider the deterministic adversary $\A$ that on input $M_{\#Q}(\db) = (y_0,\dots,y_m)$ outputs the predicate
$$p(\row) = q_0(\row) \land \left(\bigwedge_{i = 1}^m \bigl(\row[i] = y_i\bigr)\right).$$
Observe that $\iso{q_0}{\db}\implies\iso{p}{\db}$ and that by  construction $\weight_{U_m}(p) = 2^{-m}$. Thus
\begin{align*}
  \Succ_{\le2^{-m}}^{\A,M_{\#Q}}(n)
  & = \Pr_{\substack{\db\gets U_m^n \\ p\gets \A(M_{\#Q}(\db))}}[\iso{p}{\db}] \\
  & \ge \Pr_{\substack{\db\gets U_m^n \\ p\gets \A(M_{\#Q}(\db))}}[\iso{q_0}{\db}]\\
  & \ge B(n,1/n) - \negl(n)
  \tag*{\qedhere}
\end{align*}
\end{proof}

\begin{remark}
When the attack succeeds, all the predicates $q_i$ match 0 or 1 rows in $\db$. It may seem that an easy way to counter the attack is by masking low counts, a common measure taken e.g., in contingency tables. However, it is easy to modify the attack to only use predicates matching $\Theta(n)$ rows using one extra query. This means that restricting the mechanism to suppress low counts cannot prevent this type of attack.
Let $q^*$ be a predicate with $\weight_{U_m}(q^*) = 1/2$ (e.g., parity of the bits), and let $q_i^* = q_i \lor q^*$.  The attack succeeds whenever $q^*(\db) = q_0^*(\db) + 1$. If $q^*(\row)$ and $q_0(\row)$ are independent, then this occurs with probability at least $\frac{1}{2}\cdot B(n,1/n) - \negl(n)$. As before, the probability can be amplified to $1-\negl(n)$.
\end{remark}

While a single count is PSO secure for \emph{any} data distribution, the above attack against $\omega(\log(n))$ counts applies only to the uniform distribution $U_m$.
Using the Leftover Hash Lemma, we can generically extend the attack to general distributions $D$ with moderate min-entropy, at the cost of randomizing the attacked mechanism (i.e., set of counts).
Informally, we hash the data to a smaller domain where its image will be almost uniformly distributed, and adapt the attack appropriately. See Appendix~\ref{app:composition} for details.

Theorem~\ref{thm:count-composition-uniform} can be extended to the predicate mechanism $M_Q$; this follows from the observation that $M_{\#Q}$ can be implemented by postprocessing $M_Q$.
But in fact a much stronger attack is possible.
\begin{claim}
	\label{claim:leakage-composition-uniform}
  For a collection of predicates $Q = (q_1,\dots,q_m)$, let $M_{Q}(\db) \triangleq (M_{q_1}(\db),\dots,M_{q_m}(\db))$.
  Let $\univ = \zo^m$ and $D = U_m$ the uniform distribution over $\univ$.
  For $m = \omega(\log(n))$, there exists $Q$ and an adversary $\A$ such that $\A$ fully predicate singles out against $M_Q$ and $D$.
\end{claim}

\begin{proof}[Proof Outline]
For $i \in [m]$, define the predicate $q_i(\row) = \row[i]$, the $i$th bit of $\row$. Let $Q_\bits = (q_1,\dots,q_m)$.
For each row $j\in[n]$ and column $i\in[m]$, $M_{Q_\bits}(\db)$ outputs the bit $\row_i[j]$.
The adversary outputs the collection of predicates $\{p_j\}_{j\in[n]}$ where $$p_j(x) = \bigwedge_{i=1}^m \bigl(\row[i] = \row_j[i]\bigr).\qedhere$$
\end{proof}


\subsubsection{Failure to compose twice}
\label{sec:failcomposetwice}

\renewcommand{\l}{{\frac{n}{2}}}

Borrowing ideas from~\cite{NSSSU18}, we construct two mechanisms $M_\ext$ and $M_\enc$ which are individually secure against singling out (for arbitrary distributions), but which together allow an adversary to single out with high probability when the data is uniformly distributed over the universe $\univ = \zo^m$.
With more work, the composition attack can be extended to more general universes and to distributions with sufficient min-entropy.

We divide the input dataset into three parts: a source of randomness $\db_\ext = (\row_1,\ldots,\row_{\l})$, a message $\msg$, and a holdout set $\db_\hold = (\row_{\l+1},\ldots,\row_{n-1})$ used in the proof.
$M_\ext(\db)$ outputs an encryption secret key $\sk$ based on the rows in $\db_\ext$, using the von Neumann extractor.
\begin{mechanism}[H]
		\caption{$M_\ext$}
		\SetKw{KwBy}{by}
		\SetKwInOut{Input}{input}
		\Input{$\db$}
		\BlankLine
		$\sk \gets \emptyset$, the empty string\;
		\For{$i \gets 1$ \KwTo $\l$ \KwBy $2$}{
			\uIf{$\lsb(\row_i) = 0 \land \lsb(\row_{i+1}) = 1$}{$\sk \gets s\|0$}
			\uIf{$\lsb(\row_i) = 1 \land \lsb(\row_{i+1}) = 0$}{$\sk \gets s\|1$}
		}
		\uIf{$|\sk| \ge m$}{return $\sk[1:m]$, the first $m$ bits of $\sk$}
		\uElse{return $\bot$}
\end{mechanism}

$M_\enc(\db)$ runs $\sk\gets M_\ext$. If $\sk \neq \bot$, it outputs $\sk\oplus \msg$ (using $\sk$ as a one-time pad to encrypt $\msg$); otherwise, it outputs $\bot$.
Alone, neither $\sk$ nor $\sk\oplus \msg$ allows the adversary to single out,
but using both an adversary can recover $\msg$ and thereby single it out.

\begin{theorem}
	\label{thm:two-composition}
	$M_\ext$ and $M_\enc$ are secure against predicate singling out (Definition~\ref{def:security-against-singling-out}). For $m = \omega(\log(n))$ and $m\le n/8$, $\univ =\zo^m$, and $D = U_m$ the uniform distribution over $\univ$, there exists an adversary $\A$ such that
	$$\Succ_{\le2^{-m}}^{\A,M_\ee}(n) \ge 1- \negl(n),$$
	where $M_\ee = (M_\ext,M_\enc)$.
\end{theorem}

\begin{proof}
Let $\db = (\db_\ext, \db_\hold,\row_n)$ as described above.

\begin{proof}[Security of $M_\ext$.]
This is a special case of the security of the predicate mechanism $M_q$ (Theorem~\ref{thm:mech-one-bit}) and post-processing, with $q = \lsb$.
\footnotemark

\footnotetext{\label{foot:uniformity-1}
The security of $M_\ext$ does not follow from the mere fact that its output is nearly uniform. For example, the mechanism that outputs $\row_1$ may be uniform, but it trivially allows singling out. Security would follow if the output was nearly uniform \emph{conditioned} on $\db$.}

In fact, $M_\ext$ is even $(\ln(2),0,1/n,1/n)$-PSO secure. We provide a brief outline of the proof. Consider a related mechanism $M_\ext^\top$ that outputs $\top$ if $|\sk| \ge m$ and $\bot$ otherwise.
By Proposition~\ref{lemma:mech-small-codomain}, $M_\ext^\top$ is $(\ln(2),0,1/n,1/n)$-PSO secure. The security of $M_\ext$ can be reduced to that of $M_\ext^\top$ using a generalization of Proposition~\ref{lemma:permutation} to distributions of permutations.
\end{proof}

\begin{proof}[Security of $M_\enc$.]
For $\A$, $w_\low(n) < \negl(n)$, and $w_\high(n) =\omega(\log(n)/n)$,  let
$$	\gamma_\low = \Succ_{\le\w_\low}^{\A,M_\enc}(n)
\quad\mbox{and}\quad
	\gamma_\high = \Succ_{\ge\w_\high}^{\A,M_\enc}(n). $$
We must show that $\gamma_\low, \gamma_\high < \negl(n)$. It is easy to bound $\gamma_\high$ using the holdout set $\db_\hold$, which is independent of the output $M_\enc$:
$$\gamma_\high
\le \Pr_{\db,M_\enc,\A}[p(\db_\hold) \le 1 \mid \weight_{D}(p) \ge \w_\high]
= (1-\w_\high)^{n-m-2}
= o(1-\log(n)/n)^{\Omega(n)} = \negl(n).$$

\noindent
To bound $\gamma_\low$, we consider the two possible values of $p(\row_n)$. Write $\gamma_\low = \gamma_\low^0 + \gamma_\low^1$ where
\begin{align*}
\gamma_\low^b
\triangleq \Pr\left[\iso{p}{\db} \land \weight_{D}(p) \le w_\low \land p(\row_n) = b \right]
\end{align*}

If $\A$ singles out and $p(\row_n) = 1$, then the it must have gleaned information about $\row_n$ from the ciphertext $\sk \oplus \row_n$, which should be impossible.
The von Neumann extractor guarantees that either $\sk = \bot$ or $\sk$ is uniformly distributed in $\zo^\l$. Either way, the output of $M_\enc(\db)$ is information-theoretically independent of $\msg$. Therefore
$$\gamma_\low^1 \le \Pr_{\db,M_\enc,\adv}[p(\msg) = 1 \mid \weight_{D}(p) \le \w_\low] \le w_\low=\negl(n).$$

If $\A$ singles out and $p(\row_n) = 0$, then it is effectively singling out against the sub-dataset $\db_{-n} = (\row_1,\dots,\row_{n-1})$.
That is
\begin{align*}
\gamma_\low^0
&=\Pr_{\db,M_\enc,\adv}[\iso{p}{\db} \land \weight_{D}(p) \le w_\low \land p(\msg) = 0] \\
&= \Pr_{\db,M_\enc,\adv}[\iso{p}{\db_{-n}} \land \weight_{D}(p) \le w_\low \land p(\msg) = 0]
\end{align*}
We construct $\B$ that tries to single out against mechanism $M_\ext$ using $\adv$.
We assume that $\B$ can sample from $D$.\footnotemark\
On input $\sk$, $\B$ samples $\msg'\sim D$ and runs $p\gets\adv(\sk\oplus\msg')$.
\begin{align*}
\Succ_{\le\w_\low}^{\B,M_\ext}(n)
&\ge \Pr\left[\iso{p}{\db_{-n}} \land \weight_{D}(p) \le w_\low \land p(\row_n') = 0 \land p(\row_n) = 0\right] \\
&\ge \Pr\left[\iso{p}{\db_{-n}} \land \weight_{D}(p) \le w_\low \land p(\row_n') = 0\right]\cdot \Pr[p(\row_n) = 0 \mid \weight_{D}(p) \le \w_\low] \\
&\ge \gamma_\low^0 \cdot (1 - \w_\low) \\
&\ge \gamma_\low^0 \cdot (1 - \negl(n))
\end{align*}
Therefore $\gamma_\low^0$ is negligible.

\footnotetext{\label{foot:uniformity-2}
	It is tempting to try to remove this assumption by picking $\msg'$ arbitrarily, say $\msg' = 0^m$. Because $\sk$ is uniform, the ciphertexts $\sk\oplus \msg$ and $\sk\oplus \msg'$ are identically distributed and perfectly indistinguishable. This intuition is misleading (see also footnote~\ref{foot:uniformity-1}).}
\end{proof}

\begin{proof}[Insecurity of $M_\ee$ for $D = U_m$.]
The output of $M_\ee(\db)$ is a pair $(\sk, \ct)$.
If $(\sk,\ct) = (\bot, \bot)$, $\A$ aborts.
The for-loop in $M_\ext$ extracts $n/4$ uniform bits in expectation. By a Chernoff Bound, for $m\le n/8$, $\Pr_\db[\sk = \bot] \le e^{-n/16} = \negl(n)$.

If $(\sk,\ct) \neq (\bot, \bot)$,
$\A$ recovers $\msg = \ct\oplus \sk$ and outputs the predicate
$$p(\row) = \bigl(\row = \msg\bigr).$$
By the choice of $m = \omega(\log(n))$,
$\weight_{U_m}(p) = 2^{-m} < \negl(n)$.
$\Pr[\iso{p}{\db} \mid \sk \neq \bot] = 1 - \Pr[\exists j\neq n: \row_j =\row_n] = 1 - n\cdot 2^{-m} > 1- \negl(n)$.
The bound on $\Succ_{\le2^{-m}}^{\A,M_\ee}$ follows, completing the proof of the claim and the theorem.
\end{proof}
\end{proof}

\subsubsection{Singling out and failure to compose}

The failure to compose demonstrated in Section~\ref{sec:failcomposelog} capitalizes on the use of multiple counting queries. Such queries underlie a large variety of statistical analyses and machine learning algorithms. We expect that other attempts to formalize security against singling out would also allow counting queries. If so, our negative composition results may generalize beyond the notion of PSO security.

The failure to compose demonstrated in Section~\ref{sec:failcomposetwice} is more contrived. We expect that other attempts to formalize security against singling out would allow mechanisms like $M_\ext$, where the output is uniform even conditioned on the input.
It is less clear to us whether a mechanism like $M_\enc$ would be allowed under other possible formalizations of security against singling out. If an alternate formalization is to compose, it likely must forbid $M_\enc$.


\section{Differential Privacy, generalization and PSO security}
\label{sec:generalization}

\subsection{Preliminaries from differential privacy}

For $\db, \db' \in \univ^n$, we write $\db\sim\db'$ if the two datasets differ on exactly one element $\row_i$.

\begin{definition}[Differential Privacy~\cite{DMNS06, dwork2006our}] A randomized mechanism $M: \univ^n \rightarrow T$ is $(\epsilon,\delta)$-differentially private if for all $\db\sim\db'\in \univ^n$ and for all events $S\subseteq T$,
$$\Pr[M(\db)\in S] \leq e^\epsilon \Pr[M(\db')\in S] + \delta,$$ where the probability is taken over the randomness of the mechanism $M$.
\end{definition}

\begin{lemma}[Basic and Parallel Composition~\cite{dwork2006our, mcsherryPINQ}]
	Let $M$ $(\epsilon,\delta)$-differentially private and $M'$ $(\epsilon',\delta')$-differentially private. The mechanism $M' \circ M:\db \mapsto M'(M(\db),\db)$ is $(\epsilon + \epsilon', \delta+\delta')$-differentially private.
	Let $(\db_1,\dots,\db_\ell)$ be a partition of $\db$ into disjoint datasets.
	The mechanism $M^\ell:(\db_1,\dots,\db_\ell) \mapsto (M(\db_1),\dots,M(\db_\ell))$ is $(\epsilon,\delta)$-differentially private.
\end{lemma}

\begin{theorem}[Exponential Mechanism~\cite{mcsherry2007mechanism}\label{thm:expmech}]
	For domain $\univ^n$ and outcome space $\mathcal{R}$, let $u:\univ^n \times \mathcal{R} \to \mathbb{R}$ be a utility function. The sensitivity of $u$ is $\Delta u = \max_{r\in\mathcal{R}} \max_{\db\sim\db'} |u(\db,r) - u(\db',r)|$.
	For a dataset $\db$, let $\opt_u(\db) = \max_{r\in\mathcal{R}} u(\db,r)$ and let $\mathcal{R}_{\opt} = \{r\in \mathcal{R} : u(\db,r) = \opt_u(\db)\}$.
	For any $\eps > 0$, there exists a mechanism $M_\mathsf{Exp}^{\eps}:\univ^n\times \mathcal{R}\times u \to \mathcal{R}$
	that is $(\eps,0)$-differentially private such that for all $\db$ and all $t>0$:
	$$\Pr\left[
		u(M_\mathsf{Exp}^{\eps}(\db,u,\mathcal{R}) \le \opt_u(\db) - \frac{ 2\Delta u}{\eps}
		\left( \ln \left( \frac{|\mathcal{R}|}{|\mathcal{R}_{\opt}|}
		\right) +t \right)
	\right] \le e^{-t}.$$
\end{theorem}

Our analysis of how PSO security relates to differential privacy is through a connection of both concepts to statistical generalization. For differential privacy, this connection was established in~\cite{DFHPRR15, BNSSSU16}. We will also use a variant of the latter result  from~\cite{NS2015}:\footnote{The proof of Equation~\ref{eq:generalization1} of Lemma~\ref{lemma:generalization} is identical to that of Lemma~3.3 in~\cite{NS2015}, skipping the last inequality in the proof. The proof of Equation~\ref{eq:generalization2} is analogous.}
\begin{lemma}[Generalization lemma]
	\label{lemma:generalization}
	Let $\adv:(\univ^n)^\ell \to 2^\univ \times [\ell]$ be an $(\epsilon, \delta)$-differentially private algorithm that operates on $\ell$ sub-datasets and outputs a predicate $p:\univ\to\zo$ and an index $i \in [\ell]$. Let $\vec{\db} = (\db_1,\dots,\db_\ell)$ where every $\db_i\sim D^n$ is a dataset containing i.i.d. elements from $D$, and let $(p,i) \gets \adv(\vec{\db})$. Then
	\begin{eqnarray}
\E_{\vec{\db} \sim (D^n)^\ell} \left[\E_{(p,i)\gets\adv(\vec{\db})}\left[p(\db_i)\right]\right]
& \le &
e^\epsilon\cdot\E_{\vec{\db} \sim (D^n)^\ell} \left[\E_{(p,i)\gets\adv(\vec{\db})}\left[\weight_{D}(p)\right]\right] + \ell\delta \label{eq:generalization1} \\
\E_{\vec{\db} \sim (D^n)^\ell} \left[\E_{(p,i)\gets\adv(\vec{\db})}\left[p(\db_i)\right]\right]
& \ge &
e^{-\epsilon}\left(\E_{\vec{\db} \sim (D^n)^\ell} \left[\E_{(p,i)\gets\adv(\vec{\db})}\left[\weight_{D}(p)\right]\right] - \ell\delta\right).  \label{eq:generalization2}
\end{eqnarray}

\end{lemma}

\subsection{Differential privacy implies PSO security}

\begin{theorem}
	\label{thm:generalization}
For all $\eps = O(1)$, $\delta = \negl(n)$, $\w_\low\le 1/n$, and $\w_\high(n) = \omega(\log n / n)$, if $M$ is $(\eps,\delta)$-differentially private, then $M$ is $(\eps',\delta',w_\low,w_\high)$-PSO secure for
$$
\eps' = \eps + (n-1)\ln\left(\frac{1}{1-\w_\low}\right)
\quad\mbox{and}\quad
\delta' = \negl(n).
$$
\end{theorem}
\noindent
For $\w_\low = o(1/n)$, $\eps' = \eps + o(1)$.\footnote{For all $\w_\low\le 1/n$ and $n$, $\eps' < \eps + 1$ by the fact that $(1-w_\low)^{n-1} \ge (1-1/n)^{n-1} > e^{-1}$.}

\begin{proof}
The theorem consists of Claims~\ref{claim:generalization-low} and~\ref{claim:generalization-heavy}, each using one part of the generalization lemma.
That lemma holds even when the distribution $D$ is known, a fact used in both proofs.

\begin{claim}
	\label{claim:generalization-low}
	If $M$ is $(\epsilon,\delta)$-d.p., then for all $\adv$ and $\w_\low\in [0,1/n]$
	$$\Succ^{\A,M}_{\le\w_\low}(n) \le e^{\eps'} \cdot \base(n,\w_\low) + n\delta.$$
\end{claim}

\begin{claim}
	\label{claim:generalization-heavy}
	For $\eps = O(1)$ and $\delta = \negl(n)$, if $M$ is $(\epsilon,\delta)$-d.p.,
	then for all $\adv$ and all $\w_\high = \omega(\log n /n)$,
	$$\alpha\triangleq\Succ^{\A,M}_{\ge\w_\high}(n) \le \negl(n).$$
\end{claim}

\begin{proof}[Proof of Claim~\ref{claim:generalization-low}]
	Let $\w^* = \max\{w\le \w_\low : w \mbox{ realizeable under } D\}$.
	Given $p\gets\adv(M(\db))$, $\w_\low$, and $D$, define the predicate $p^*$:
	$$p^*(x) \equiv \begin{cases}
	p(x)& \mbox{if } \weight_{D}(p) \le \w_\low \\
	0& \mbox{if } \weight_{D}(p) > \w_\low
	\end{cases}$$
	Observe that $\weight_{D}(p^*) \le \w^*$.
	The predicate $p^*$ can be computed from $p$, $D$, and $\w_\low$ without further access to $\db$.
	Because differential privacy is closed under post-processing, if $M$ is $(\eps,\delta)$-differentially private, then the computation that produces $p^*$ is as well.
	\begin{align*}
	\Succ_{\le\w_\low}^{\A,M}(n)
	&\le \Pr_{\db,p}[p(\db) \ge 1/n \land \weight_{D}(p) \le \w^*]    \\
	&\le n \cdot \E_{\db,p} [p^*(\db)]  \\
	&\le n\cdot(e^\epsilon \w^* + \delta) \;\;\,\quad\quad \quad\quad\quad \mbox{by Lemma~\ref{lemma:generalization}, $\ell=1$}  \\
	&= e^\eps \frac{\base(n,w^*)}{(1-w^*)^{n-1}} + n\delta \quad\quad\quad \mbox{by Claim~\ref{claim:baseline-exact}}   \\
	&\le e^\eps \frac{\base(n,w^*)}{(1-w_\low)^{n-1}} + n\delta   \\
	&= e^{\eps'} \base(n,w_\low) + \delta' \;\quad \quad\quad\mbox{by Claim~\ref{claim:baseline-exact}} \tag*{\qedhere}
\end{align*}
\end{proof}

\begin{proof}[Proof of Claim~\ref{claim:generalization-heavy}]
	Fix an adversary $\adv$.
	We construct an algorithm $\B$ in an attempt to violate the Generalization Lemma for $\ell = O(\frac{\log n}{\alpha})$:
	\begin{adversary}[H]
			\SetKwInOut{Input}{Input}
			\SetKw{KwBy}{by}
			\Input{$D$, $\vec{\db}\sim (D^n)^\ell$}
			\BlankLine
			$I \gets \emptyset$, the empty set\;
			\For{$i\gets 1,\dots,\ell$}{
				$p_i \gets \adv(M(\db_i))$\;
				$u_i = -p_i(\db_i)$\;
				\uIf{$\weight_{D}(p_i) \ge w_\high$}{$I \gets I\cup\{i\}$}
			}
			Let $u:i\mapsto -p_i(\db_i)$ for $i\in I$\;
			$i^* \gets M_\mathsf{Exp}^{\eps}(\vec{\db},I,u)$\;
			return $(i^*, p_{i^*})$
	\end{adversary}
	$M$ is $(\eps,\delta)$-differentially private, and $M_\mathsf{Exp}^\eps$ (Theorem~\ref{thm:expmech}) is $(\eps,0)$-differentially private. By basic and paraellel composition, $\B$ is $(2\eps,\delta)$-differentially private.

	\newcommand{\SO}{\mathsf{PSO}}
	\newcommand{\pbad}{(1 - \alpha)^\ell}
	Define the event $\SO$ to be the event that $\A$ successfully predicate singles out on one of the sub-datasets with a high-weight predicate: $\SO = \{\exists i \in [\ell] : \iso{p_i}{\db_i} \land \weight_{D}(p_i) \ge \w_\high\}$.
	By the choice of $\ell$, $\Pr[\SO] = 1- \left(1 - \alpha\right)^{\ell} \ge 1-\frac{1}{n}$.
	Conditioned on $\SO$, $\max_{i\in I} u(i) \ge -1/n$.
	$\Delta u = 1/n$, and $|I| \le \ell$. The Exponential Mechanism guarantees that
	$$\Pr_{\vec{\db}; (p_{i^*},i^*) \gets \B(\vec{\db})}\biggl[p_{i^*}(\db_{i^*}) \ge \frac{1}{n} + \frac{2}{n\eps}(\ln\ell + t) \mid\SO\biggr]\le e^{-t}.$$
	Choosing $t = \ln n$ and using the fact that  $p_{i^*}(\db_{i^*}) \le 1$,
	\begin{align*}
		\E_{\vec{\db}; (p_{i^*},i^*) \gets \B(\vec{\db})}[p_{i^*}(\db_{i^*})\mid\SO]
		&\le \frac{1}{n} + \frac{2}{n\eps}(\ln \ell + \ln n) + \frac{1}{n}.\notag
	\end{align*}

	\begin{align*}
	\E_{\vec{\db}; (p_{i^*},i^*) \gets \B(\vec{\db})}\left[p_{i^*}(\db_{i^*})\right]
	&= \Pr[\neg\SO]\E[p_{i^*}(\db_{i^*})\mid\neg\SO] + \Pr[\SO] \E[p_{i^*}(\db_{i^*})\mid\SO] \notag\\
	&\le \Pr[\neg\SO]+ \E[p_{i^*}(\db_{i^*})\mid\SO] \notag\\
	&<  \frac{3}{n} + \frac{2}{n\eps}(\ln \ell + \ln n).
	\\ \\
	%
	\E_{\vec{\db}; (p_{i^*},i^*) \gets \B(\vec{\db})}\left[\weight_{D}(p_{i^*})\right]
	&= \Pr[\neg\SO]\E[\weight_{D}(p_{i^*})\mid\neg\SO] + \Pr[\SO] \E[\weight_{D}(p_{i^*})\mid\SO] \notag \\
	&\ge \Pr[\SO]\cdot\E[\weight_{D}(p_{i^*})\mid\SO] \notag\\
	&\ge (1-\frac{1}{n})\cdot\w_\high\notag\\
	&> \frac{3\w_\high}{4}
\end{align*}
	Applying the Lemma for the $(2\eps,\delta)$-d.p. mechanism $\B$,
	$$\frac{3}{n} + \frac{2}{n\eps}(\ln \ell + t)
	\ge
	e^{-2\eps}\left(\frac{3\w_\high}{4} - \ell\delta\right).$$
	If $\delta = \omega(\frac{\alpha}{n})$, then by the assumption that $\delta$ is negligible, $\alpha = \negl(n)$.
	Otherwise $\delta = O(\frac{\alpha}{n})= O(\frac{\log n}{n\ell})$ and
	$$\frac{2}{\eps}(\ln \ell + \ln n) \ge \frac{3n\w_\high- O(\log n)}{4e^{2\eps}}.$$
	For $\eps = O(1)$ and $\w_\high = \omega(\frac{\log n}{n})$, $\ln \ell + \ln n 	 =\omega(\log n)$. By the choice of $\ell = O(\frac{\log n}{\alpha})$,  $\alpha = \negl(n)$.
\end{proof}

\end{proof}


\section{Does $k$-anonymity provide PSO security?}
\label{sec:k-anon}

$k$-anonymity~\cite{SamaratiS98, sweeney2002k} is a strategy intended to help a data holder ``release a version of its private data with scientific guarantees that the individuals who are the subjects of the data cannot be re-identified while the data remain practically useful''~\cite{sweeney2002k}. It is achieved by making each individual in a data release indistinguishable from at least $k-1$ individuals. Typically, a $k$-anonymized dataset is produced by subjecting it to a sequence of generalization and suppression operations.

The Article~29 Working Party Opinion on Anonymisation Techniques concludes that $k$-anonymity prevents singling out~\cite{wp-anonymisation}. In this section, we analyze the extent to which $k$-anonymity provides PSO security.
We show that $k$-anonymized dataset typically provides an attacker information which is sufficient to predicate singling out with constant probability.
This result challenges the determination of the Article~29 Working Party.\footnote{Our results hold equally for $\ell$-diversity~\cite{MKGV07} and $t$-closeness~\cite{LiLV07} which the Article~29 Working Party also concludes prevent singling out.}

\subsection{Preliminaries}
\newcommand{\ahat}{\widehat{a}}

Let $(A_1, \dots,A_m)$ be \emph{attribute domains}.
A dataset $\db=(\row_1,\ldots,\row_n)$ is collection of rows $\row_i =(a_{i,1},\ldots,a_{i,m})$ where $a_{i,j}\in A_j$.
For subsets $\ahat_{i,j} \subseteq A_j$, we view $y_i = (\ahat_{i,1},\dots,\ahat_{i,m})$ as a set in the natural way, writing $\row_i \in y_i$ if $\forall j\in [m]$, $a_{i,j} \in \ahat_{i,j}$.
We say that a dataset $\dby = (y_1,\dots,y_n)$ is derived from $\db$ by generalization and suppression if $\forall i\in [n]$, $\row_i \in y_i$.
For example, if $(A_1, A_2,A_3)$ correspond to ``5 digit ZIP Code,'' ``Gender,'' and ``Year of Birth,'' then it may be that $\row_i=(91015, F, 1972)$ and $y_i = (91010\text{--}91019, F, 1970\text{--}1975)$.

$k$-anonymity aims to capture a sort of anonymity of a crowd: a data release $\dby$ is $k$-anonymous if any individual row in the release cannot be distinguished from $k-1$ other individuals. Let $\cnt(\dby,y) \triangleq |\{i\in[n]: y_i = y\}|$ be the number of rows in $\dby$ which agree with $y$.\footnotemark

\footnotetext{Often $\cnt$ is paramaterized by a subset $Q$ of the attribute domains called a \emph{quasi-identifier}. This parameterization affect our analysis and we omit it for simplicity.}

\begin{definition}[$k$-Anonymity (rephrased from~\cite{sweeney2002k})]
For $k\ge2$, a dataset $\dby$ is \emph{$k$-anonymous} if $\cnt(\dby,y_i) \ge k$ for all $i\in[n]$.
An algorithm is called a \emph{$k$-anonymizer} if on an input dataset $\db$ its output is a $k$-anonymous $\dby$ which is derived from $\db$ by generalization and suppression.
\end{definition}

Our goal is to relate $k$-anonymity and PSO security.
It will be convenient to define a generalization of $k$-anonymity---\emph{predicate $k$-anonymity} which captures the core property of $k$-anonymity but relaxes its strict syntactic requirements.

For a predicate $\phi:\univ \to \zo$ and dataset $\db$, let $\db_\phi = \{\row \in \db : \phi(\row) = 1\}$.
We assume that $|\db_\phi|$ is computable given the output of the $k$-anonymizer, but this does not qualitatively affect the results in this section.
\begin{definition}[\label{def:predicate_anonymity}Predicate $k$-anonymity]
	Let $\anon$ be an algorithm mapping a dataset $\db \in \univ^n$ to a collection of predicates $\Phi = \{\phi:\univ\to\zo\}$.
	For $k\ge 2$ we call $\anon$ \emph{predicate $k$-anonymous} if for all $\phi \in \Phi$, $|\db_\phi| \ge k$.
\end{definition}

$k$-anonymity is a special case of predicate $k$-anonymity that considers only specific collections of predicates $\Phi$ induced by a dataset $\dby$:
$$\Phi = \{ \phi_y(\row) = 1 \iff \row \in y\}_{y\in\dby}.\footnotemark$$

\footnotetext{See also the definition of $k$-anonymity for face images \cite[Definition~2.10]{k-faces}. Using the notation of that paper, it is also special case of predicate $k$-anonymity, with $\Phi = \{\phi_{\Gamma_d}(\Gamma) = 1 \iff f(\Gamma) = \Gamma_d \}_{\Gamma_d \in H_d}$}

\begin{definition}
A predicate $k$ anonymizer is \emph{$\kmax$-bounded} if $\forall \db, \exists \phi \in \Phi$ such that $|\db_\phi| \le \kmax$.
\end{definition}

\subsection{Illustrative examples}

Before presenting a formal technical analysis, we provide two illustrative examples of very simple $k$-anonymizers that fail to provide security against predicate singling out. For both examples, let $D = U_\ell$ be the uniform distribution over $\zo^n$. The dataset $\db$ consists of $n$ records sampled i.i.d.\ from $D$.

\paragraph{Bit suppression.}
This $k$-anonymizer processes groups of $k$ rows in index order and suppresses all bit locations where the $k$ rows disagree.
Namely, for each group $g$ of $k$ rows $(\row_{gk+1},\dots,\row_{gk+k)})$ it outputs $k$ copies of the string $y_g \in \{0,1,\star\}^n$ where $y_g[j] = b\in\{0,1\}$ if $\row_{gk+1}[j]=\cdots=\row_{gk+k}[j]=b$ (i.e., all the $k$ rows in the group have $b$ as their $j$th bit) and $y_g[j] = \star$ otherwise.

In the terminology of Definition~\ref{def:predicate_anonymity}, the predicate $\phi_g(x)$ evaluates to $1$ if $y_g[j] \in \{x[j], \star\}$ for all $j\in[n]$ and evaluates to $0$ otherwise. Namely, $\phi_g(x)$ checks whether $x$ agrees with $y_g$ (and hence with all of $\row_{gk+1},\dots,\row_{gk+k)}$) on all non-suppressed bits.

In expectation, $n/2^{k}$ positions of $y_g$ are not suppressed.
For large enough $n$, with high probability over the choice of $\db$, at least $\frac{n}{2\cdot 2^k}$ positions in $y_g$ are not suppressed.
In this case, $\weight_{D}(\phi_g) \le 2^{-\frac{n}{2\cdot 2^k}}$ which is a negligible function of $n$ for any constant $k$.

We now show how $\phi_g$ can be used adversarially. In expectation $n(1-2^{-k})\geq 3n/4$ positions of $y_g$ are suppressed. For large enough $n$, with high probability over the choice of $\db$ at least $n/2$ of the positions in $y_g$ are suppressed. Denote these positions $i_i,\ldots,i_{n/2}$.
Define the predicate $p_k(x)$ that evaluates to $1$ if the binary number resulting from concatenating $x[i_1],x[i_2],\ldots,x[i_{n/2}]$ is greater than $2^{n/2}/k$ and $0$ otherwise. Note that  $\weight_{D}(p_k)\approx 1/k$ and hence $p_k$ isolates within the group $g$ with probability $\approx 1/e\approx 0.37$, as was the case with the trivial adversary described at the beginning of Section~\ref{sec:definitions}.

An attacker observing $\phi_g$ can now define a predicate $p(x)=\phi_g(x)\wedge p_k(x)$. By the analysis above, $\weight(p)$ is negligible (as it is bounded by $\weight(\phi_g)$) and $p(x)$ isolates a row in $\db$ with probability $\approx 0.37$. Hence, the $k$-anonymizer of this example fails to protect against singling out.

Theorem~\ref{thm:k-anon-attack} below captures the intuition from our bit suppression example and generalizes it, hence demonstrating that $k$-anonymity would not typically protect against predicate singling out. We note that Theorem~\ref{thm:k-anon-attack} does not capture all possible ways in which the outcome of a $k$-anonymizer can be exploited, in particular, the following simple example.

\paragraph{Interval Buckets.}
This $k$-anonymizer sorts the rows in lexicographic order and outputs the intervals $[a_g,b_g] = [x_{gk+1},x_{gk+k}]$ (where the indices are after sorting and renaming).
The corresponding predicate $\phi_{a_g,b_g}(x) = 1$ if $x\in [a_g,b_g]$.

Observe that any of the endpoints $a_g$ or $b_g$ reveal a row in $\db$ and hence an adversary can predicate single out with probability 1 using predicates of weight $2^{-n}$.

\subsection{k-Anonymity enables predicate singling out}

\begin{theorem}
	\label{thm:k-anon-attack}
For any $\kmax \ge 2$,
there exists an (efficient, uniform, randomized) algorithm $\adv$ such that for all
$D$ with min-entropy  $\lambda \ge m + 2\log(1/\alpha^2) + \kmax\log n$ (for $m \in \N$, $\alpha$), and all predicate anonymizers $\anon$ that are $\kmax$-bounded,
and all $\w_\low > 0$:
	$$
	\Succ_{\le\w_\low}^{\A,\anon}(n)
		\ge
	\eta\cdot (e^{-1} - 2^{-m}n - k\alpha^2)	$$
where
	$$\eta \triangleq  \Pr_{\substack{\db\gets D^n\\\phi\gets\anon(\db)}}\left[\weight_{D}(\phi) \le w_\low(n)\right].$$
\end{theorem}

For distributions with sufficient min-entropy ($m = \omega(\log n)$, $\alpha = \negl(n)$), the adversary's success probability is approximately $\eta/e \approx \eta \cdot B(k,1/k)$.
To predicate single out, the adversary must output a predicate that both isolates $\db$ and has low weight.
The theorem shows that these two requirements essentially decompose: $\eta$ is the probability that the predicate $k$-anonymizer outputs a low-weight predicate and $B(k,1/k)$ is the probability that a trivial adversary predicate singles out a dataset of size $k$.
Algorithms for $k$-anonymity generally try to preserve as much information in the dataset as possible. We expect such algorithms to typically yield low-weight predicates and correspondingly high values of $\eta$.

\begin{proof}[Proof of Theorem~\ref{thm:k-anon-attack}]
	On input $\Phi\gets \anon(\db)$, $\adv$ selects $\phi \in \Phi$ such that $2\le |\db_\phi| \le \kmax$.
	$\adv$ will construct some predicate $q$ and output the conjunction $p\triangleq \phi \land q$.
	Noting that $\weight_{D}(p) \le \weight_{D}(\phi)$, and that $\iso{q}{\db_\phi} \implies \iso{p}{\db},$
\begin{align}
\Succ_{\le\w_\low}^{\A,\anon}(n)
&\ge \Pr\left[\iso{q}{\db_\phi} \land \weight_{D}(\phi) \le w_\low\right] \notag \\
&= \eta \cdot \Pr\left[\iso{q}{\db_\phi}  \bigmid \weight_{D}(\phi)\le w_\low\right] \label{eqn:k-anon-proof-cond-1}
\end{align}
\begin{claim}
\label{eqn:k-anon-proof-cond}
There exists $\adv$ such that for all $k_\phi\ge 2$
\begin{equation*}
\Pr_{\substack{\db\gets D^n\\\phi\gets\anon(\db)\\p\gets\adv(\phi,w)}}\biggl[\iso{q}{\db_\phi} \Bigmid |\db_\phi| = k_\phi  \land \weight_{D}(\phi)\le\w_\low\biggr] \ge B(k_\phi,1/k_\phi) - 2^{-m}n - k\alpha^2.
\end{equation*}
\end{claim}
\noindent
The claim is discussed below and proved in Appendix~\ref{app:k-anon}. Using the claim we get:
\begin{align}
\Pr\left[\iso{q}{\db_\phi} \bigmid \weight_{D}(\phi)\le \w_\low\right]
&= \sum_{k_\phi = k}^{\kmax} \Pr\left[|\db_\phi| = k_\phi\right] \cdot \Pr\left[\iso{q}{\db_\phi} \bigmid |\db_\phi| = k_\phi \land \weight_{D}(\phi)\le \w_\low\right] \notag \\
&\ge \sum_{k_\phi} \Pr\left[|\db_\phi| = k_\phi\right] \cdot \biggl( B(k_\phi,1/k_\phi) - 2^{-m}n -  k\alpha^2  \biggr) \notag \\
&= \E_{k_\phi} \left[B(k_\phi,1/k_\phi)\right] - 2^{-m}n - k\alpha^2 \notag \\
&\ge e^{-1} - 2^{-m}n - k\alpha^2 \label{eqn:k-anon-proof-449}
\end{align}
The last inequality follows from the fact that for all $k_\phi \ge 2$, $(1-1/k_\phi)^{k_\phi-1} > e^{-1}$.
Combining~\eqref{eqn:k-anon-proof-449} with~\eqref{eqn:k-anon-proof-cond-1} completes the proof.
\end{proof}

The proof of Claim~\ref{eqn:k-anon-proof-cond} uses the Leftover Hash Lemma in a manner closely resembling Lemma~\ref{lemma:lhl-1}, but with an additional challenge. The earlier application of LHL proved that a random hash function selected appropriately isolates a row with probability close to $e^{-1}$. It relied on the fact that each row was sampled i.i.d.\ from a distribution with min-entropy. In contrast, the rows in $\db_\phi$ are a function of $\anon$ and the whole dataset $\db$. They are not independently distributed and even their marginal distributions may be different than $D$.

We can use the LHL to prove the claim if we can show that the rows in $\db_\phi$ still have sufficient (conditional) min-entropy. The following lemma (proved in Appendix~\ref{app:k-anon}) does exactly that.
\begin{lemma}
	\label{lemma:k-of-n-min-entropy}
Let $Y_1,\dots,Y_n$ be i.i.d.\ random variables and let $F$ be a (randomized) function mapping $(Y_1,\dots,Y_n)$ to $(j,I)$ where $j\in[n]$ and $I\subseteq [n]\setminus\{j\}$ of size $|I| = k-1$. Let $Y_I = \{Y_i\}_{i\in I}$.
$$\widetilde{H}_\infty(Y_j \bigmid Y_I) \ge {H}_\infty(Y_j) - (k-1)\log n \ge {H}_\infty(Y_1) - k\log n.$$
\end{lemma}


\section*{Acknowledgment}

The authors thank Uri Stemmer for discussions of the generalization properties of differential privacy and Adam Sealfon for suggesting Proposition~\ref{lemma:mech-small-codomain}.

Work supported by the U.S.\ Census Bureau under cooperative agreement no.~CB16ADR0160001. Any opinions, findings, and conclusions or recommendations expressed in this material are those of the authors and do not necessarily reflect the views of the U.S.\ Census Bureau. Aloni Cohen was additionally supported by NSF award CNS-1413920, the 2018 Facebook Fellowship, and MIT's RSA Professorship and Fintech Initiative.

\bibliographystyle{alpha}
\bibliography{bib}

\appendix

\section{Omitted proofs from Section~\ref{sec:definitions}}
\label{app:definitions}

\begin{proof}[Proof of Claim~\ref{claim:baseline-exact}]
	For $w \in [0,1]$, let $P_w = \{p:\weight_D(p) = w\}$.
	First, we show that for all $\w$, $ \base(n,P_\w) \le B(n,\w).$
	For any fixed predicate $p$, $$\Pr_{\db \sim D^n}[\iso{p}{\db}] = {n \choose 1} \cdot \weight_D(p) \cdot (1 - \weight_D(p))^{n-1} = B(n, \weight_D(p)).$$
	For a trivial adversary $\T$, let $\alpha_\w(\T) = \Pr_{\T(\bot)}[\weight_D(p) = \w]$.
	\begin{equation}
		\label{eqn:base-alpha}
		\base(n,P_\w) = B(n,w)\cdot \sup_{\mbox{{\scriptsize Trivial }} \T} \alpha_w(\T)
	\end{equation}
	If $w$ is realizable under $D$, then there exists a deterministic trivial adversary $\T$ with $\alpha_w(\T) = 1$; otherwise $\alpha_w(\T) = 0$ for all $\T$.

	For $W \subseteq [0,1]$, let $P_W = \{p:\weight_D(p) \in W\}$. By definition, $\base(n,P_W) \ge \sup_{\w \in W} \base(n,P_\w).$ Next, we show that in fact $\base(n,P_W) = \sup_{\w \in W} \base(n,P_\w).$
	Suppose, towards contradiction, that $\base(n,P_W) > \sup_{\w\in W} \base(n,P_w)$.
	Then there exists trivial $\T$ with $\Succ_{P_W}^{\T,\bot}(n) > \sup_{\w\in W} \base(n,P_w).$
	There must also exist an deterministic trivial adversary $\T'$ with $\Succ_{P_W}^{\T',\bot}(n) \ge \Succ_{P_W}^{\T,\bot}(n)$ but which always outputs predicates of a single weight $\w'\in W$, a contradiction.

	Combining with \eqref{eqn:base-alpha}: for any $W$, $$\base(n,P_W) = \sup_{\substack{w\in W \\\mbox{\scriptsize realizable}}} B(n,w).$$
	Because $B(n,w)$ monotonically increases as $w\to 1/n$,
	\begin{align*}
		\sup_{\substack{w\le \w_\low \\\mbox{\scriptsize realizable}}} B(n,w) &= B(n,\w^*_\low) \\
		\sup_{\substack{w\ge \w_\high \\\mbox{\scriptsize realizable}}} B(n,w) &= B(n,\w^*_\high) \\
	\end{align*}
\end{proof}

\begin{proof}[Proof of Lemma~\ref{lemma:lhl-1}]
	We prove the lemma for $\w\ge2^{-(m-1)}$; the proof for $\w \le 1 - 2^{-(m-1)}$ is analogous.
	Identify the set $\zo^m$ with the set $\{0,1,\dots,2^m -1\}$ in the natural way. For $y \in \zo^m$, define the function $r(y) \triangleq \frac{y}{2^m-1}$, the projection of $y$ onto the interval $[0,1]$.
	Let $0 \le \Delta \le \w$
	be some constant to be chosen later, and let $w_m$ be the greatest multiple of $2^{-m}$ less or equal to  $w-\Delta$.
	\begin{align}
		\Pr_{y \in_R \zo^{m}}[r(y) \le w-\Delta]
		&= \Pr_{y \in_R \zo^{m}}[r(y) \le w_m]  \notag \\
		&= w_m + 2^{-m} \notag\\
		&\in [w-\Delta,w-\Delta+2^{-m}]\label{eqn:pr-r-le-w}
	\end{align}
	Let $H = \{h:\univ\to\zo^m\}$ be 2-universal family of hash functions.
	For each $h\in H$ we define the predicate $p_h$:
	$$p_h(\row) = \begin{cases}
	1 & r(h(\row)) \le \w - \Delta \\
	0 & r(h(\row)) > \w - \Delta
	\end{cases}.$$
	By the Leftover Hash Lemma, for every $\alpha > 0$ to be chosen later and every $\lambda \ge m + 2 \log(1/\alpha^2)$, if $D$ has min-entropy at least $\lambda$ then
	$$\bigl(h,h(\row)\bigr)_{\substack{h\in_R H \\ \row \sim D}}$$
	is $\alpha^2$-close to the uniform distribution over $H\times \univ_n$ in total variation distance.
	By Corollary~\ref{cor:sd-fact} $h(D)$ is $\alpha$-close to uniform over $\univ$ with probability at least $1-\alpha$ over $h\in_R H$. For such $h$, by \eqref{eqn:pr-r-le-w},
	\begin{equation*}
		\label{eqn:good-hash-property}
		\weight_D(p_h) = \Pr_{\row\gets D}[p_h(\row) \le \w-\Delta] \in \bigl[\w-\Delta - \alpha
	\mbox{\textbf{,}}\,\,\, 
	\w-\Delta+2^{-m}+\alpha\bigr].
\end{equation*}
	Set $\alpha = 2^{-m}$ and $\Delta = 2\alpha$.
	Then $\weight_D(p_h) \in[\w - 3\cdot 2^{-m}, \w]$ with probability at least $1- \alpha = 1-2^{-m}$ whenever $\lambda \ge m + 2\log(1/\alpha^2) = 5m$, completing the proof.
\end{proof}


\section{Omitted proofs from Section~\ref{sec:composition}}
\label{app:composition}

Given a 2-universal family of functions $H = \{h:\univ \to \zo^m\}$, a mechanism $M$ and adversary $\A$, we construct a new randomized mechanism $M_H$ and new adversary $\A_H$.
The following lemma relates the success probability of the modified $\A_H$ with respect to $D$ to that of $\A$ with respect to $U_m$.

\begin{mechanism}[H]
		\caption{$M_H:\univ \to Y$}
    \SetKwInOut{Output}{fixed}
    \Output{$H = \{h:\univ \to \zo^m\}$, and $M:\zo^m \to Y$ }
		\SetKwInOut{Input}{input}
		\Input{$\db \in \univ$}
		\BlankLine
    sample $h \in_R H$\;
    return $(h, M(h(\db)))$
\end{mechanism}
\noindent

\begin{lemma}
  \label{lemma:unif-to-general}
  For any $\A$ there exists $\A_H$ such that for all $M$, $\w_\low$, $\w_\high$, $\alpha > 0$ and $D\in\Delta(\univ)$ with min-entropy $\lambda > m+2\log(1/\alpha^2)$:

  \begin{align*}
  \left|\Succ^{\A_H,M_H}_{\le\w_\low+\alpha}(n,D) - \Succ^{\A, M}_{\le\w_\low}(n,U_m)\right| \le n\alpha
  \end{align*}
  where $\Succ(n,D)$ (respectively, $\Succ(n,U_m)$) denotes the PSO success probability with respect to the distribution $D$ (respectively, $U_m$) as in Definition~\ref{def:success}.
\end{lemma}

\begin{proof}
  For a predicate $p$ on $\zo^m$ we define a corresponding predicate on $\univ$: $p_h(x) \triangleq p(h(x))$.
  On input $(h,y)\gets M_H(\db)$, $A_H$ simulates $p\gets A(y)$ and outputs $p_h$.

  We call $h\in H$ \emph{good} if $h(D)$ is $\alpha$-close to $U_m$.
  By Corollary~\ref{cor:sd-fact}, $(1-\alpha)$-fraction of $h$ are good.
  By the goodness of $h$, if $\weight_{U_m}(p) \le \w_\low$ then $\weight_D(p_h) \le \w_\low + \alpha$. 
	$$
	\left| \Succ_{\le \w_\low+\alpha}^{\A_H,M_H}(n, D)
	-\Succ_{\le\w_\low }^{\A,M}(n,U_m)\right| \le  \SD(U_m^n, h(D^n)) \le n\alpha
  \qedhere
  $$
\end{proof}

\begin{corollary}
	\label{thm:count-composition}
  Let $\alpha = \negl(n)$ and $\lambda = m+2\log(1/\alpha^2)$.
  For any $m = \omega(\log(n))$, there exists a distribution over $O(m)$-many predicates $Q_h$, a negligible function $\w_\low(n)$, and an adversary $\A$ such that for all $D$ with min-entropy at least $\lambda$: $$\Succ_{\le \w_\low}^{\A,M_{\#Q_h}}(n) \ge 1 - \negl(n).$$
\end{corollary}
\begin{proof}
  The success probability in Theorem~\ref{thm:count-composition-uniform} is easily amplified from $1/e$ to $1-\negl(n)$ by repetition.
  Applying Lemma~\ref{lemma:unif-to-general} to the result almost completes the proof; it remains to verify that the resulting mechanism $M_H$ can be written as $M_{\#Q_h}$ for some $Q_h = (q_0^h,\dots,q_m^h)$. To do so, take $q_i^h(\row) = q_i(h(\row))$, where $q_i$ is from the proof of Theorem~\ref{thm:count-composition-uniform}.
\end{proof}
\begin{remark}
Corollary~\ref{thm:count-composition} is only meaningful as an example of a failure of composition if each $M_{\#q_i^h}$ taken in isolation is PSO secure, something that is \emph{not} provided by Lemma~\ref{lemma:unif-to-general}. However, $M_{\#q_i^h}$ is an instance of the counting mechanism and thus secure.
\end{remark}


\section{Omitted proofs from Section~\ref{sec:k-anon}}
\label{app:k-anon}

\begin{proof}[Proof of Lemma~\ref{lemma:k-of-n-min-entropy}]
	We prove the second inequality first. The idea in used in~\eqref{eq:entropy-444} is used in~\eqref{eq:entropy-444-2}.
	\begin{align}
	2^{-{H}_\infty(Y_j)}
	&= \max_{y} \Pr[Y_j = y]  \notag \\
	&\le \max_y \Pr[\exists \ell \in [n], Y_\ell = y ]  \label{eq:entropy-444} \\
	&\le n\cdot\max_y \Pr[Y_1 = 1]  \notag \\
	&= 2^{\log n-{H}_\infty(Y_j)} \notag
	\end{align}
	The first inequality:
	\begin{align}
	2^{-\widetilde{H}_\infty(Y_j \mid Y_I)}
	&= \E_{Y_I} \biggl[\max_{y} \Pr[Y_j = y \Bigmid Y_I]\biggr]  \notag \\
	&= \sum_{y_I} \Pr[Y_I = y_I] \cdot \left(\max_{y} \Pr[Y_j = y \Bigmid Y_I = y_I] \right) \notag \\
	&= \sum_{y_I}  \max_{y} \biggl(\Pr[Y_I = y_I] \cdot\Pr[Y_j = y \Bigmid Y_I = y_I] \biggr) \notag \\
	&= \sum_{y_I} \max_{y} \biggl(\Pr[Y_j = y]\cdot \Pr[Y_I = y_I \Bigmid Y_j = y] \biggr) \notag \\
	&\le \sum_{y_I} \max_{y} \biggl(\Pr[Y_j = y]\cdot \Pr[\forall i\in I, \exists\ell \in [n]\setminus\{j\}, Y_\ell = y_i \Bigmid Y_j = y] \biggr) \label{eq:entropy-444-2} \\
	&= \sum_{y_I} \max_{y} \biggl(\Pr[Y_j = y]\cdot \Pr[\forall i\in I, \exists\ell \in [n]\setminus\{j\}, Y_\ell = y_i] \biggr) \notag \\
	&= \sum_{y_I} \Pr[\forall i\in I, \exists\ell \in [n]\setminus\{j\}, Y_\ell = y_i] \cdot \max_{y} \Pr[Y_j = y] \notag \\
	&= \max_{y} \Pr[Y_j = y] \cdot \sum_{y_I} 1 \notag \\
	&\le {n\choose k-1}\cdot 2^{-{H}_\infty(Y_j)} \notag \\
	&\le 2^{(k-1)\log n - {H}_\infty(Y_j)} \notag
	\end{align}
\end{proof}

The proof of Claim~\ref{eqn:k-anon-proof-cond} applies the following corollary of the Leftover Hash Lemma.  For random variables $Y_1,\dots,Y_k$, and $j \in [n]$, let $Y_{-j} = \{Y_i : i\neq j\}$.
\begin{corollary}[Corollary to Leftover Hash Lemma (\ref{LHL})]
  \label{LHL:avg-min-entropy}
  For every $Y_1,\dots,Y_k$ if $\forall j\in[n]$, $\widetilde{H}_\infty(Y_j \Bigmid Y_{-j}) = \lambda \ge m + 2\log(1/\alpha^2)$, then $(h(Y_1), \dots, h(Y_k))_{h\in_R H}$ is $k\alpha^2$-close to uniform over $\left(\zo^m\right)^k$ in total variation distance.
\end{corollary}
\begin{proof}
\end{proof}

\begin{proof}[Proof of Claim~\ref{eqn:k-anon-proof-cond}]
The construction of $q$ uses the Leftover Hash Lemma and is very similar to the construction of the predicates in Lemma~\ref{lemma:lhl-1}.
Identify the set $\zo^m$ with the set $\{0,1,\dots,2^m -1\}$ in the natural way. For $y \in \zo^m$, define the function $r(y) \triangleq \frac{y}{2^m-1}$, the projection of $y$ onto the interval $[0,1]$.

Let $w_\phi$ be the multiple of $2^{-m}$ closest to $1/k_\phi$.
Observe that $|B(k_\phi,w_\phi) - B(k_\phi,1/k_\phi)| \le \left|w_\phi - \frac{1}{k_\phi}  \right| \cdot \max_{\w' \in [0,1]} \left|\dv{B}{w}(w')\right| \le 2^{-m}n$.

Let $H = \{h:\univ\to\zo^m\}$ be a 2-universal family of hash functions.
For each $h\in H$ define the predicate $q_h$:
$$q_h(\row) = \begin{cases}
1 & r(h(\row)) < w_\phi \\
0 & r(h(\row)) \ge w_\phi
\end{cases}.$$
Because $w_\phi$ is a multiple of $2^{-m}$,
\begin{align*}
	\Pr_{y \in_R \zo^{m}}[r(y) < w_\phi] = w_\phi
\end{align*}

By Lemma~\ref{lemma:k-of-n-min-entropy}, $\db_{\phi}$ (viewed as a $k_\phi$-tuple of random variables) satisfies the average min-entropy hypothesis of Corollary~\ref{LHL:avg-min-entropy}. Applying that Corollary:
\begin{align*}
\Pr_{\db_{\phi},q_h}\biggl[\iso{q}{\db_\phi} \Bigmid |\db_\phi| = k_\phi  \land \weight_D(\phi)\le\w_\low\biggr]
&\ge \Pr_{y_1,\dots,y_{k_\phi} \in_R \zo^m}[\exists \mbox{ unique } j\in[k_\phi] : r(y_j) < w_\phi] - k_\phi \alpha^2\\
&= B(k_\phi,w_\phi) - k_\phi \alpha^2 \\
& \ge B(k_\phi,1/k_\phi) - 2^{-m}n - k_\phi \alpha^2.
\end{align*}
\end{proof}


\end{document}